\documentclass[3p]{elsarticle}
\usepackage{geometry} 
\usepackage{amsfonts}
\usepackage{mathtools}
\usepackage{amsthm}
\usepackage{amsmath}
\usepackage{mathdots}
\usepackage{bbm}
\usepackage{hyperref} 
\usepackage{cleveref}
\usepackage{color}
\usepackage{subcaption}
\usepackage{natbib}
\usepackage{multirow}
\setcitestyle{authoryear, open={[},close={]}}
\usepackage[linesnumbered,ruled,vlined]{algorithm2e}

\usepackage{tikz}
\usetikzlibrary{arrows.meta}
\usetikzlibrary{shapes} 

\usepackage{array}
\usepackage{booktabs} 

\usepackage[symbol]{footmisc}

\input{Settings/Commands}
\title{Pricing and hedging the prepayment option of mortgages under stochastic housing market activity\texorpdfstring{\footnote[2]{The views expressed in this paper are the personal views of the authors and do not necessarily reflect the views or policies of their current or past employers. The authors have no competing interests.}}{}}

\author[1]{Leonardo Perotti\texorpdfstring{\corref{cor1}}{}}
\ead{L.Perotti@uu.nl}
\author[1,2]{Lech A.~Grzelak}
\ead{L.A.Grzelak@uu.nl}
\author[1]{Cornelis W.~Oosterlee}
\ead{C.W.Oosterlee@uu.nl}
\cortext[cor1]{Corresponding author.}
\address[1]{Mathematical Institute, Utrecht University, Utrecht, the Netherlands}
\address[2]{Financial Engineering, Rabobank, Utrecht, the Netherlands}

\date{\today}

\begin{document}

\begin{abstract}
    \noindent Prepayment risk embedded in fixed-rate mortgages forms a significant fraction of a financial institution's exposure. The embedded prepayment option bears the same interest rate risk as an exotic interest rate swap with a suitable stochastic notional. Focusing on penalty-free prepayment because of the contract owner's relocation to a new house, we model the prepayment option value as an European-type interest rate receiver swaption with stochastic maturity matching the stochastic time of relocation. This is a convenient representation since it allows us to compute the prepayment option value in terms of well-known pricing formulas for European-type swaptions.
    We investigate the effect of a stochastic housing market activity as the explanatory variable for the distribution of the relocation time, as opposed to the conventional assumption of a deterministic housing market activity. We prove that the housing market covariance drives the prepayment option price difference between the stochastic setting and its deterministic counterpart.
    The prepayment option exposure is hedged using market instruments based on Delta-Gamma replication. Furthermore, since the housing market activity is a non-tradable risk factor, we perform non-standard actuarial hedging focusing on controlling the prepayment option exposure yield by risky housing market scenarios. 

\end{abstract}


\maketitle


\section{Introduction}
\label{sec: Intro}

Mortgages constitute one of the main exposures on a financial institution's balance sheet. Hence, a thorough analysis of their risks is crucial. Unrealistic default-free, non-callable fixed-interest rate mortgages bear \emph{interest rate risk} since banks' funding is indexed to (floating) market rates. However, the risk is neutralized by a suitable position in interest rate swaps (IRSs). A different story applies when we consider real mortgage contracts that generally allow for counterparty default and a counterparty's call of some embedded optionality. The most significant among such optionalities is the right to partially or totally repay the outstanding mortgage notional in advance, i.e. the \emph{embedded prepayment option} (EPO). A scrupulous modeler should account for these features while pricing (and hedging!) a portfolio of mortgages. Failing to do so, may lead to a significant mismatch between the expected and the realized exposures.

With a focus on the Dutch mortgage market, however, the occurrence of a counterparty default is so rare as to be considered negligible. Hence, the attention can be dedicated to the assessment of the EPO.\footnote{This would not be true for other mortgage markets, such as the US market. In that framework, both default and prepayment are required to be addressed.} On this topic extensive research has been conducted over the years, leading to two main classes of models: \emph{financially rational models} (sometimes also called \emph{option theoretic models}) and \emph{empirical models}. Financially rational models assume a mortgage owner exercises the prepayment option optimally as soon as the cost of continuing with the original mortgage is higher than the cost of entering into a new mortgage repriced at the prevailing market rates, see, among others, \citep{dunn1981comparison,stanton1995rational}. The main criticality of financially rational models is the assumption of optimal exercise from a usually non-financially educated counterparty, often leading to an overestimate of the actual EPO value. This is also the main justification for empirical models that use historical data to forecast prepayment behavior. In such a framework, the model aims to exploit client- and contract-specific features to predict prepayment behavior. This approach requires a deep understanding and inference of the explanatory variables for the prepayment occurrences, and a rich literature exists \citep{alink2002mortgage,charlier2003prepayment,kalotay2004option,hoda2007implementation,lin2010determinants,caspari2018valuation}. Within the spectrum of prepayment models, we find also \emph{extended exogenous models} that augment the financially rational models including exogenous/empirical features in the EPO payoff formulation. Recent examples of this hybrid class are \citep{casamassima2022pricing,perotti2024modeling}, where the EPO exercise is a function of market risk factors, however, the magnitude of the prepayment is computed based on data-driven relationships.

A peculiar aspect of the Dutch regulation concerns mortgage prepayment in case of relocation to a new house by the mortgage \emph{owner}, sometimes called \emph{borrower}. In particular, the owner has two viable options: fully prepaying the outstanding notional penalty-free or ``taking'' the old mortgage to the new house (``porting option'') \citep[see][]{hassink2011importance}. While penalty-free prepayment in case of relocation is common practice in several countries, the porting option is specific to the Dutch mortgage market. Not surprisingly, a financially rational owner will exercise the porting option if the prevailing mortgage rate is higher than the fixed rate locked in by the old mortgage; on the other hand, the owner will prepay when the prevailing mortgage rate is lower than the old contractual fixed rate. In other words, at the relocation time (RT), a mortgage owner has the right (but not the obligation) to reprice the old mortgage at the prevailing fixed interest rate.

Since most people who decide to relocate rely on the expert judgment of a market professional (e.g., the majority of people go through a real estate agency to assist them in the process) it is realistic to assume the decision to prepay or bring the old mortgage to the new house is taken according to financially rational criteria. Hence, we aim to represent the price of the \emph{embedded prepayment option due to relocation} (EPOR) with a European-type payoff, where the exercise time coincides with the RT, and so it is stochastic.

The EPOR value is then dependent on two main market drivers. The level of the market rates is responsible for the option being in or out of the money, while the housing market's (HM) activity affects the RT distribution. Given these two components, the valuation of the EPOR is straightforward. 

\subsection{Contribution}

The research in this paper focuses on the impact of the housing market on the valuation of the EPO. Assuming a rational exercise of the prepayment-porting option, the moment a relocation to a new house is observed, we use standard option pricing theory to value prepayment because of relocation, conditional on the relocation occurrence. Hence, an important building block is the description of the relocation distribution itself. We start from the work of \citep{lando1998cox}, where a Cox process represents the default time distribution. The same machinery is employed here to describe the unknown relocation time distribution. The underlying variable is, in our framework, the \emph{housing market activity}, i.e. the fraction of transactions observed in the housing market.

We investigate the impact on the EPO value of a stochastic future housing market, as opposed to the choice for a simpler model that assumes a deterministic housing market. We provide a convenient pricing formula that requires the evaluation of the \emph{expected relocation density} by means of a basic numerical integral. We assess the effect of different choices for the stochastic driver of the housing market activity. Eventually, we perform hedging of the EPO exposure based on matching the market instruments ``Greeks,'' such as its Delta and Gamma. As a second hedging experiment, we take a risk-management perspective and build a replicating strategy that protects the financial institution against extreme realizations of the housing market.

The remainder of the paper is organized as follows. The model is detailed in \Cref{sec: methodology}. Particularly, \Cref{ssec: PayoffPrepayment} presents the EPO payoff; \Cref{ssec: ModelingRelocation} is dedicated to the modeling of the relocation distribution using a Cox process and the pricing of the EPO; in \Cref{ssec: NonlinearAdjustment}, we study the main factors that affect the relocation value compared to a deterministic model. \Cref{sec: Hedging} describes the general hedging problem. In \Cref{ssec: ExplainableHedge}, we propose a hedging strategy bearing intuitive economic meaning; whereas \Cref{ssec: ActuarialHedge} addresses the need for a robust hedge under housing market activity uncertainty. In \Cref{sec: DataCalibration}, we illustrate the method employed for the relocation density calibration. \Cref{sec: Experiments} is devoted to the numerical experiments. \Cref{sec: Conclusion} provides a summary of the research conducted and concludes.

\section{Methodology}
\label{sec: methodology}

The purpose of this paper is to define the prepayment price for a fixed-rate mortgage contract due to a borrower's relocation to a new house as an EU-type swaption with stochastic maturity equal to the random time of relocation.

\subsection{Payoff of the prepayment option at time of relocation}
\label{ssec: PayoffPrepayment}

We first observe that fixed-rate mortgages can be considered -- from an interest rate risk perspective -- as amortizing interest rate swaps (IRSs) with notional and tenor characteristics consistent with the contractual mortgage. This representation will be convenient when defining the payoff of the prepayment option in case of relocation in terms of the payoff of well-known market instruments.
\begin{ass}[Equivalence between fixed rate mortgages and IRSs]
\label{ass: EquivalenceSwapMortgage}
    Let $N$ represent the mortgage contractual amortization scheme, i.e. $N(t)$ represents the outstanding notional at time $t$, and let $K$ be the mortgage contractual fixed rate. In a frictionless market, by coupling the fixed-rate mortgage with a suitable set of floating rate notes(FRNs) -- that always trade at par \citep{brigo2006interest} --  the financial institution's net exposure is equivalent to an amortizing IRS exposure, with notional schedule $N(t)$ and fixed-rate $K$ \citep[see, for example,][]{perotti2024modeling}.
\end{ass}

We represent the value at time $T$ of a receiver IRS with starting date $t_0$, payment dates $\T_{\tt{p}}=\{t_1,\dots,t_{n}\}$, $t_j>t_0$, $j=1,\dots,n$, and strike $K$, as follows:
\begin{equation}
\label{eq: IRSPayoff}
    S(T,K) = A(T)(K-\kappa(T)),\qquad t_0\leq T<t_n,
\end{equation}
with IRS annuity factor, $A(T)$, and prevailing interest rate, $\kappa(T)$, respectively defined as:
\begin{align}
\label{eq: Annuity}
    A(T)&\equiv A(T;\T_{\tt{p}},N)=\sum_{\substack{t_j\in\T_{\tt{p}}\\ t_j> T}}N(t_{j-1})(t_j-t_{j-1}) P(T;t_{j}),\\
    \label{eq: SwapRate}
    \kappa(T)&\equiv \kappa(T;\T_{\tt{p}},N)=\frac{1}{A(T)}\sum_{\substack{t_j\in\T_{\tt{p}}\\ t_j> T}}N(t_{j-1})\big(P(T;t_{j-1})-P(T;t_j)\big),
\end{align}
for $P(T; S)$ the price at time $T$ of a zero coupon bond (ZCB) with maturity $S\geq T$. Observe that, for the sake of simplicity, we write explicitly the dependence of the value on time $T$ and fixed rate $K$ only. $\T_{\tt{p}}$ and $N(t)$ are, indeed, fixed.

\begin{figure}[t]
    \centering
    \resizebox{0.9\textwidth}{!}{
    \begin{tikzpicture}

\draw[->, ultra thick] (-1.3,0.5) -- (11.5,0.5) node[anchor=north] {};

\draw[dashed, ultra thick] (4.2,0.5) -- (4.2,4.8) node[anchor=south] {};

\node[circle, fill=black, inner sep=2pt, label=below:{$t_0$}] at (-1., 0.5) {};
\node[circle, fill=black, inner sep=2pt, label=below:{$T^*$}] at (11., 0.5) {};
\node[circle, fill=black, inner sep=2pt, label=below:{$\tau$}] at (4.2, 0.5) {};

\node at (1.5,4.5) {exposure at $T \leq \tau$};
\node at (7.5,4.5) {exposure at $T > \tau$};

\node[draw, rectangle, fill=green!20, minimum width=3.2cm, minimum height=1.3cm, very thick] (leftbox) at (1.,2.5) {$A(T)(K-\kappa(T))$};

\node[draw, rectangle, fill=green!20, minimum width=3.2cm, minimum height=1.3cm, very thick] (rightbox1) at (9.,3.5) {$A(T)(K-\kappa(T))$};
\node[draw, rectangle, fill=red!20, minimum width=3.2cm, minimum height=1.3cm, very thick] (rightbox3) at (9.,1.5) {$A(T)(\kappa(\tau)-\kappa(T))$};

\node[draw, ellipse, fill=yellow!20, minimum size=0.25cm] (relocbox) at (4,1.5) {};

\draw[-, thick] (leftbox.east) -- (4,3.5) -- (rightbox1.west);
\draw[-, thick] (leftbox.east) -- (4,1.5) -- (rightbox1.west);
\draw[-, thick] (relocbox.east) -- (6,1.5) -- (rightbox3.west);

\node[draw, ellipse, fill=yellow!20, minimum size=0.25cm, thick] at (4.2,3.5) {\footnotesize no reloc.};
\node[draw, ellipse, fill=yellow!20, minimum size=0.25cm, thick] (relocbox) at (4.2,1.5) {\footnotesize reloc.};
\node[draw, ellipse, fill=cyan!20, minimum size=0.25cm, thick] at (5.9,2.6) {\footnotesize port.};
\node[draw, ellipse, fill=cyan!20, minimum size=0.25cm, thick] at (6.1,1.5) {\footnotesize prep.};

\end{tikzpicture}}
    \caption{\footnotesize Scheme of the exposure evolution given the occurrence of relocation, porting, and prepayment.}
    \label{fig: RelocationPortingPrepayment}
\end{figure}

\begin{figure}[b]
    \centering
    \resizebox{0.9\textwidth}{!}{
    \begin{tikzpicture}

\draw[->, ultra thick] (-0.3,-0.3) -- (12.5,-0.3) node[anchor=north] {};

\node[draw, rectangle, fill=green!20, minimum width=12.cm, minimum height=1.5cm, very thick] at (6.0,4.25) {$A(T)(K-\kappa(T))$};

\node[draw, rectangle, fill=green!20, minimum width=4.cm, minimum height=3.5cm, very thick] at (2.0,1.75) {$A(T)(K-\kappa(T))$};

\node[draw, rectangle, fill=green!20, minimum width=8.cm, minimum height=1.75cm, very thick] at (8.0,2.625) {$A(T)(K-\kappa(T))$};

\node[draw, rectangle, fill=red!20, minimum width=8.cm, minimum height=1.75cm, very thick] at (8.0,0.875) {$A(T)(\kappa(\tau)-\kappa(T))$};

\node at (0.1,4.75) [right] {\footnotesize $T\leq T^*< \tau$: no relocation};
\node at (0.1,3.25) [right] {\footnotesize $T<\tau\leq T^*$: before reloc.};
\node at (4.1,3.25) [right] {\footnotesize $\tau\leq T\leq T^*$, $\kappa(\tau) \geq K$: after reloc. + porting option};
\node at (4.1,1.5) [right] {\footnotesize $\tau\leq T\leq T^*$, $\kappa(\tau) < K$: after reloc. + prepayment option};

\node[circle, fill=black, inner sep=2pt, label=below:{$t_0$}] at (0.0, -0.3) {};
\node[circle, fill=black, inner sep=2pt, label=below:{$T^*$}] at (12., -0.3) {};
\node[circle, fill=black, inner sep=2pt, label={below:{$\tau$}}] at (4., -0.3) {};

\node at (6.,5.3) {Exposure at $T\in[t_0,T^*]$ for different scenarios};

\end{tikzpicture}}
    \caption{\footnotesize Scheme of the exposure evolution given the occurrence of relocation, porting, and prepayment.}
    \label{fig: RelocationNoRelocation}
\end{figure}

In the following, we illustrate why the prepayment option because of relocation can be represented as an EU-type IR swaption with maturity equal to the RT. Notation-wise, we denote the mortgage end date\footnote{In practice, we are interested in the end of the fixed-rate period that may not coincide with the end date of the mortgage.} as $T^*=t_n$ and the RT as $\tau$.

Let us take the perspective of a mortgage issuer sitting at time $t_0$, and observe how prepayment due to relocation may affect the issuer's exposure. According to \Cref{ass: EquivalenceSwapMortgage}, as long as no relocation occurs before the end date of the mortgage, i.e. $\tau> T^*$, a fixed rate mortgage generates a time-$T$ exposure equal to $S(T,K)$ for every $T\in[t_0,T^*]$. Similarly, when a relocation occurs at $\tau\leq T^*$, but the prevailing mortgage rate is not more convenient than the contractual rate, i.e. $\kappa(\tau)\geq K$, the time $T$ exposure amounts at $S(T,K)$ for every $T\in[t_0,T^*]$ since, after relocation, the mortgage owner exercises the porting option benefitting of the contractual, low fixed rate (compared to the prevailing rate). On the other hand, if the prevailing mortgage rate at the time of relocation is lower than the contractual rate, i.e. $\kappa(\tau)< K$, the mortgage owner prepays the old mortgage and ``replaces'' it with a new mortgage, thus locking in the cheaper fixed rate, $\kappa(\tau)$, for the remaining tenor. The issuer observes a time $T$ exposure of $S(T,\kappa(\tau))$ for every $T\in[\tau,T^*]$ (the exposure remains $S(T,K)$ for $T\in[t_0,\tau)$). The different possible actions from the borrower side are illustrated in \Cref{fig: RelocationPortingPrepayment}, with a focus on the resulting exposures (taking the financial institution's point of view).

Mathematically, we write the issuer exposure at time $T\geq t_0$, when relocation at time $\tau$ is considered, as:
\begin{equation*}
\begin{aligned}
    Y(T,K)\equiv Y(T,K;\tau,\kappa(\tau))=&S(T,K)\1\{t\leq T\leq T^*<\tau\} \\
    +&S(T,K)\1\{t\leq T<\tau\leq T^*\} \\
    +&S(T,K)\1\{t\leq \tau\leq T \leq T^*\}\1\{\kappa(\tau)\geq K\}\\
    +&S(T,\kappa(\tau))\1\{t\leq\tau\leq T \leq T^*\}\1\{\kappa(\tau)< K\},
\end{aligned}
\end{equation*}
where $\1\{\cdot\}$ is the indicator function and the different parts of the right-hand side term correspond to a region in the scheme given in \Cref{fig: RelocationNoRelocation}.
To isolate the EPOR exposure (or payoff), we consider the difference between $S(T,K)$, the contractual mortgage exposure, and $Y(T,K)$. We get, for $T\geq t_0$:
\begin{align*}
    Z(T,K)\equiv Z(T,K;\tau,\kappa(\tau))&=S(T,K)-Y(T,K)\\
    &=A(T)\big(K-\kappa(\tau)\big)^+\1\{t\leq \tau\leq T \leq T^*\}.
\end{align*}
In other words, a relocation event generates exposure after the date of relocation (but before the end of the mortgage) only if the prevailing mortgage rate at the RT was lower than the contractual rate, i.e. $\kappa(\tau)<K$. In particular, we observe that the payoff at the relocation time $\tau$ is given by:
\begin{equation}
\begin{aligned}
\label{eq: EPOPayoff}
    Z(\tau,\kappa(\tau)) &=A(\tau)\big(K-\kappa(\tau)\big)^+\1\{t\leq\tau\leq T^*\}\\
    &=S(\tau,K)^+\1\{t\leq \tau\leq T^*\}.
\end{aligned}
\end{equation}

\subsection{Modeling of the relocation event and pricing of the prepayment option}
\label{ssec: ModelingRelocation}

We consider the time horizon $\T=[t_0,T^*]$ and $(\Omega,\F_{\infty},\P)$ the probability space supporting an interest rate process, $r:\Omega\times \T\rightarrow \R$, a housing market process, $h:\Omega\times \T\rightarrow \R$, and a unit exponential random variable, $E:\Omega\rightarrow \R^+$, i.e. $E\sim Exp(1)$ independent of $r$ and $h$.

The process $r$ represents the \emph{short rate}. A rich literature exists regarding possible model dynamics for $r$, see, e.g., \citep{brigo2006interest}. A common choice we will adhere to in this paper is the model proposed by \citep{hull1990pricing}, which is conveniently analytically tractable while also perfectly reproducing today's information about the prevailing market yield curve.\footnote{This model belongs, indeed, to the class of \citep*{heath1992bond} models fulfilling this requirement, crucial in real applications.} 
The process $h$ describes \emph{national HM activity}, defined as the annualized instantaneous rate of transactions observed in the Dutch housing market, i.e. the \emph{frequency} of transactions. Formally, we write:
\begin{equation}
\label{eq: HMActivity}
    h(t)=\lim_{\Delta t \rightarrow 0} \frac{1}{\Delta t }\frac{NoT(t,t+\Delta t)}{NoH(t)},
\end{equation}
where $NoT(t,t+\Delta t)$ and $NoH(t)$ are the number of transactions observed between time $t$ and $t+\Delta t$ and the total number of houses available, respectively. By definition, $h(t)$ is nonnegative, and the economic meaning of $h(t)$ is straightforward, indicating the (infinitesimal) number of transactions over a unit of time, normalized by the magnitude of the housing market. We will investigate the effect of different model choices for $h$.
The remaining building block is the relocation time, $\tau$. We model the (first) relocation event by a Cox process \citep{cox1959analysis}, following the insights delineated in \citep{lando1998cox} in the context of pricing defaultable securities. The relocation time is defined as:
\begin{equation}
\label{eq: RelocationTimeTau}
    \tau=\inf\bigg\{t:\int_{t_0}^t\lambda(h(s))\d s \geq E\bigg\},
\end{equation}
for some deterministic intensity function $\lambda:\R\rightarrow\R^+$. As observed in \citep{lando1998cox}, $\tau$ is ``the first jump time of a Cox process with intensity process $\lambda(h)$.'' We observe that the Cox process is a generalization of the Poisson process where a stochastic counterpart replaces the deterministic intensity. A stochastic intensity is crucial when we believe the event modeled depends on future realizations of some underlying variable whose behavior is uncertain and can only be predicted in a distributional sense.
\begin{rem}[Relocation time probability density function]
    As a standard result for the first jump time of an inhomogeneous Poisson process, see e.g. \citep{kalbfleisch2002statistical}, for $h:\R^+\rightarrow \R^+$, the probability density function (PDF) of the relocation time $\tau$ is given by the function:
    \begin{equation}
    \label{eq: PDFTau}
    \begin{aligned}
            f^h:\:&\R^+\longrightarrow\R^+,\\ &t\longmapsto \lambda(h(t)) \e^{-\int_{t_0}^t \lambda(h(s))\d s}.
    \end{aligned}
    \end{equation}
    In general, \eqref{eq: PDFTau} forms a family of densities described by the functional $f$ defined as:
    \begin{equation}
    \label{eq: FunctionalFTau}
    \begin{aligned}
        f:\:&\big(\R^+\rightarrow \R^+\big) \longrightarrow \big(\R^+\rightarrow\R^+\big),\\ &h\longmapsto f^h.
    \end{aligned}
    \end{equation}
    In particular, whenever $h$ is a stochastic process, $f^h$ represents a family of PDFs dependent on the realization $h(\omega,\cdot)$, for $\omega\in\Omega$. Hence, for a fixed time $t$, $f^h(t)$ is a well-defined random variable with all the canonical operators, such as the expectation and variance, in place.
\end{rem}

\Cref{fig: DensityDistributions} illustrates the ``distribution'' of densities for two special cases of $h$. In \Cref{fig: DensityDistributions}a, we assume $h$ ``jumps'' at time $t_0$ to a level drawn from a normal distribution and subsequently it is constant over time. Hence, each density has a different mass at time $t_0$ because of the multiplicative term $\lambda(h(t_0))$. Then, we observe a decay over time, at different speeds, caused by the exponential damping factor $\e^{-\int_{t_0}^t \lambda(h(s))\d s}=\e^{-\lambda(h(t_0))(t-t_0)}$, see \eqref{eq: PDFTau}. \Cref{fig: DensityDistributions}b illustrates the case when $h$ starts from a given value at $t_0$ and grows linearly to $h(T^*)$, here, $T^*=10$. $h(T^*)$ is randomly drawn by the same normal distribution used for $h(t_0)$ in \Cref{fig: DensityDistributions}a. The thin gray lines represent the ``paths'' followed by each realization of $h$ over time, whereas the thick gray line is the normal density of $h(t_0)$ and $h(T^*)$, respectively.

\begin{figure}[t]
    \centering
    \subfloat[\centering]{{\includegraphics[width=7.cm]{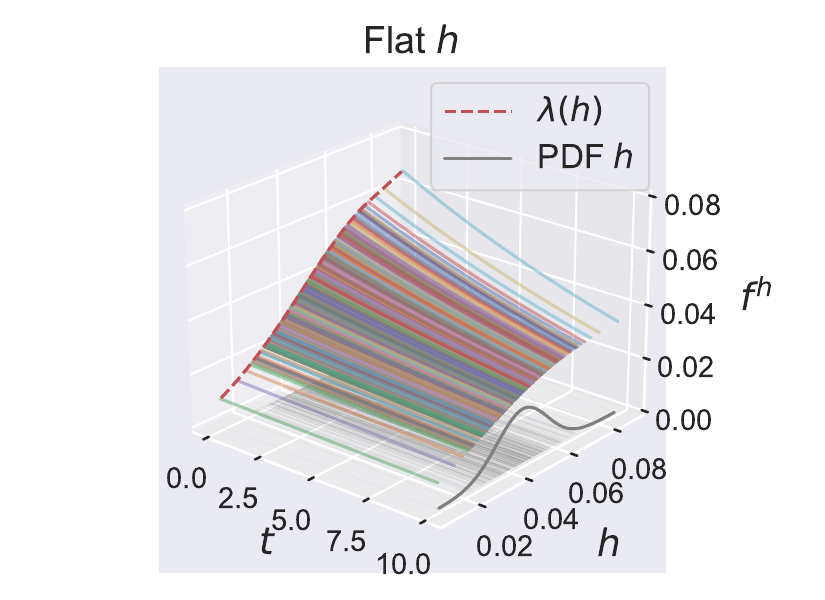} }}%
    ~\hspace{-.5cm}
    \subfloat[\centering]{{\includegraphics[width=7.cm]{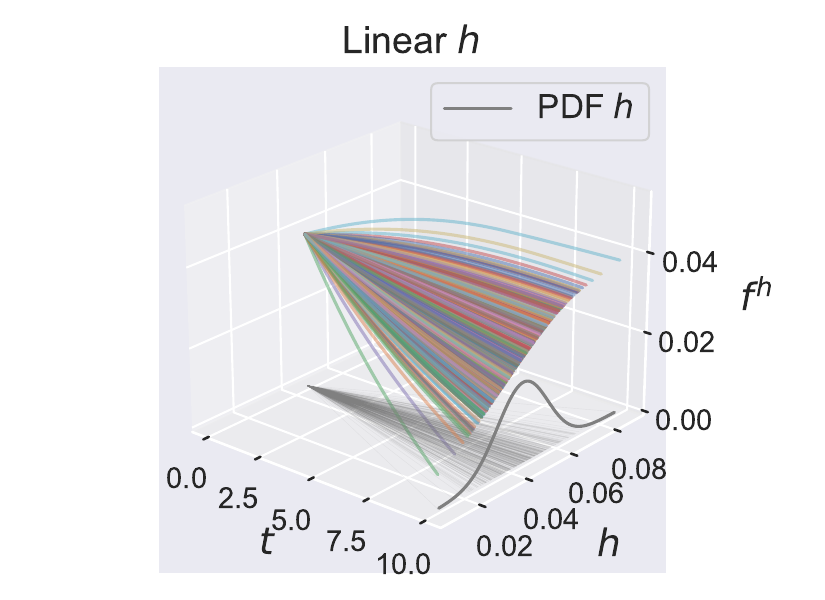} }}
    \caption{Realized relocation density, $f^h$, given different housing market activities, $h$. Left: flat $h$ in time, normally distributed at $t_0$. Right: linear $h$ in time, reaching a normally distributed final level, $h(T^*)$, for $T^*=10$.}
    \label{fig: DensityDistributions}%
\end{figure}

In addition to the hypothesis in \citep{lando1998cox}, we assume that the rate process $r$ and the housing market activity process $h$ are independent. Specifically, we assume the housing market activity $h(t)$, $t\in\T$ is independent of the current rate $r(t_0)$. Formally, we define the $\sigma$-algebras representing the amount of information for the three random quantities involved in the system as: 
\begin{equation*}
    \F_r(t)=\sigma(r(s):s\leq t),\quad \F_h(t)=\sigma(h(s):s\leq t),\quad \G(t)=\sigma(\1\{\tau\leq s\}:s\leq t).
\end{equation*}
The independence relationship stated above is formalized by:
\begin{ass}[Independence of rates and HM]
    $\F_r(t_0)\perp\F_h(t)$ for every $T\in\T$.
\end{ass}
\begin{rem}[Economic justification for independence]
\label{rem: RelocationRateIndependence}
    Empirical evidence shows that the activity in the housing market is not dependent on the prevailing interest rates. This is mainly explained by the correlation between interest rates and house prices. Low interest rates allow borrowers to apply for higher mortgages, but at the same time entail higher house prices and the other way around.  Hence, the effective ``cost of a loan'' for a specific house is ``close to invariant.'' 
\end{rem}

Following the standard derivative pricing theory, we consider the filtered measure space $(\Omega,\FF,\Q)$ where $\FF=(\F(t))_{t\in\T}$, with the increasing collection of $\sigma$-algebras, $\F(t)=\F_r(t)\vee\F_h(t)\vee\G(t)\subset \F_{\infty}$, each one representing the information available at time $t$ for the corresponding random quantity. $\Q$ is \emph{a} martingale measure, equivalent to $\P$, associated with the discount factor $D(t_0,T)=\e^{-\int_{t_0}^T r(s)\d s}$. The EPOR arbitrage-free price is the conditional expectation of the discounted exposure generated by prepayment due to relocation (see \eqref{eq: EPOPayoff}). The EPOR present value is defined as:
\begin{equation}
\label{eq: ValueEPORExpectation}
    V\equiv V(t_0) = \E^\Q\big[D(t_0,\tau)Z(\tau,\kappa(\tau))\big|\F(t_0)\big],
\end{equation}
where it is worth underlining that the swap rate $\kappa(t)$ is a measurable function of the short rate $r(t)$. Since the EPOR value is always computed at time $t_0$, we drop the explicit dependence on time in the following, using $V$ instead of $V(t_0)$. For notational convenience, we often use ``bar'' to indicate the conditional expectation w.r.t. the $\sigma$-algebra $\F_h(t_0)$, i.e. $\overline{X}\equiv \E^\Q[X|\F_h(t_0)]$ for some random variable $X\in\F_h(t_0)$. For instance, $f^h(t)$ in \eqref{eq: PDFTau} with \emph{stochastic} $h$ -- but \emph{fixed} $t$! -- is a well-defined random variable. Hence, we will write $\overline{f^h(t)}$ instead of $\E^\Q[ f^{h}(t)\big|\F_h(t_0)]$ to maintain a light notation. Observe also that the ``bar'' notation introduced here always refers to an expectation \emph{with respect to $\F_h(t_0)$}, the initial information regarding the HM activity risk factor. Hence, this notation should not be confused with the expectation in the pricing formula \eqref{eq: ValueEPORExpectation}, that is computed \emph{with respect to $\F(t_0)=\F_r(t_0)\vee\F_h(t_0)\vee\G(t_0)$}, the total initial information.

\begin{prop}[EPOR present value]
\label{prop: ValueEPORIntegral}
    Calling $\overline{f^h(T)}=\E^\Q\big[ f^{h}(T)\big|\F_h(t_0)\big]$, the EPOR present value given in \eqref{eq: ValueEPORExpectation} reads:
    \begin{equation}
    \label{eq: ValueEPORIntegral}
        V=\int_{t_0}^{T^*} C(T) \overline{f^h(T)} \d T,
    \end{equation}
    where $C(T)\equiv C(T;K;\T_{\tt{p}},N)$ is the price of a European swaption maturing at time $T$ written on a receiver (amortizing) swap with fixed rate $K$, and tenor described by (the remaining part of) $\T_{\tt{p}}$ and $N(t)$, $t\leq T$. $\overline{f^h(T)}$ is the average density, at time $T\in\T=[t_0,T^*]$, over all the possible HM activity realizations, $h$.
\end{prop}
\begin{proof}
    Observe first that $Z(t,\kappa(t))$ in \eqref{eq: EPOPayoff} is an $\F_r(t)$-adapted process, hence it is also adapted to the filtration generated by the increasing sequence $\F_{r,h}(t)=\F_r(t)\vee\F_h(t)$. Using \citep[][Proposition 3.1, eq. (3.3)]{lando1998cox} and Fubini's theorem, we get:
    \begin{equation*}
        V = \int_{t_0}^{T^*}\underbrace{\E^\Q\Big[Z(T,\kappa(T))\lambda(h(T))\e^{-\int_{t_0}^T (r(s)+\lambda(h(s)))\d s}\Big|\F_{r,h}(t_0)\Big]}_{(*)}\d T.
    \end{equation*}
    Since $\F_{r,h}(t)\subset \F_r(t)\vee\F_h(T)$, by the tower property of conditional expectations, we have:
    \begin{align*}
        (*)&=\E^\Q\bigg[\E^\Q\Big[Z(T,\kappa(T))\lambda(h(T))\e^{-\int_{t_0}^T (r(s)+\lambda(h(s)))\d s}\Big|\F_r(t_0)\vee\F_h(T)\Big]\bigg|\F_{r,h}(t_0)\bigg]\\
        &=\E^\Q\bigg[\E^\Q\Big[Z(T,\kappa(T))\e^{-\int_{t_0}^T r(s)\d s}\Big|\F_r(t_0)\Big]\lambda(h(T))\e^{-\int_{t_0}^T \lambda(h(s))\d s}\bigg|\F_{r,h}(t_0)\bigg]\\
        &=\E^\Q\Big[C(T)\lambda(h(T))\e^{-\int_{t_0}^T \lambda(h(s))\d s}\Big|\F_{r,h}(t_0)\Big]\\
        &=C(T) \E^\Q\big[ f^{h}(T)\big|\F_h(t_0)\big],
    \end{align*}
    where in the second equality we used that $f^h(T)=\lambda(h(T))\e^{-\int_{t_0}^T \lambda(h(s))\d s}\in \F_r(t_0)\vee\F_h(T)$ and that $\F_r(t_0)\perp\F_h(T)$; in the third equality we observed that the inner conditional expectation is the price of a suitable European swaption, $C(T)$; in the last equality we used that $\F_r(t_0)\perp \F_h(t_0)$. This concludes the proof.
\end{proof}

\begin{cor}[EPOR distribution w.r.t. the HM]
\label{cor: EPORDIstribution}
    The EPOR present value can be rephrased as:
    \begin{equation}
    \label{eq: EPORValueAvgDistribution}
        V=\overline{V_h}\equiv \E^\Q\big[V_h\big|\F_h(t_0)\big],
    \end{equation}
    where the random variable $V_h\equiv V_h(t_0)$ is defined as:
    \begin{equation}
        \label{eq: EPORValueDistribution}
        \begin{aligned}
    V_h:\:&\Omega\longrightarrow\R,\\ &\omega\longmapsto \int_{t_0}^{T^*} C(T) f^{h(\omega,\cdot)}(T) \d T,
    \end{aligned}
    \end{equation}
    for $f^h$ given in \Cref{eq: PDFTau}.
\end{cor}
\begin{proof}
    The result holds by definition of $V_h$, using the linearity property of the operators involved.
\end{proof}

Before moving to the next part, we list a few remarks regarding the two results above and their economic interpretation and use.
\begin{itemize}
    \item \Cref{eq: ValueEPORIntegral} in \Cref{prop: ValueEPORIntegral} provides a compact pricing formula for the EPO in case of relocation. Indeed, when a (semi)analytic pricing formula for European swaptions, $C(T)$, exists,\footnote{E.g., under the Hull-White model, vanilla and amortizing European swaptions can be priced semianalytically using Jamshidian trick. This only requires the numerical computation of the root of a suitable equation.} formula \eqref{eq: ValueEPORIntegral} requires the computation of the inner expectation (only w.r.t. process $h$) and the numerical evaluation of the outer deterministic integral. For special choices of $h$, the inner expectation allows for analytical solutions or accurate approximated solutions. In general, the two problems are modular and can be solved separately, making the evaluation computationally efficient.
    \item The average PDF for the relocation time, i.e. $\overline{f^h(T)}=\E^\Q\big[ f^{h}(T)\big|\F_h(t_0)\big]$, provides the natural weights for the relevant European swaption maturities in the representation of the EPOR value. Replication strategies based on market instruments can be developed starting from these weights.
    \item \Cref{eq: EPORValueAvgDistribution} in \Cref{cor: EPORDIstribution} states that the EPOR value is the average value considering all possible housing market scenarios. However, given the non-hedgeable nature of the HM activity, the synthesis provided by the average value is of arguable use. Indeed, since Delta-hedging is not possible, we might be interested in investigating the EPOR value distribution for different realizations of the HM, as defined in \eqref{eq: EPORValueDistribution}. In a risk management framework, the information regarding the EPOR distribution is a valuable input to develop hedging strategies based on scenario analysis (sometimes also called \emph{actuarial hedging}).
\end{itemize}

\subsection{Nonlinear adjustment and approximation of the EPOR value}
\label{ssec: NonlinearAdjustment}

A consequence of \Cref{prop: ValueEPORIntegral} is that, given the non-linearity (in $h$) of $f^h$, the \emph{true} price of the EPOR differs from the price obtained assuming the average HM realizes. This is entailed by the fact that, in general, the following holds true:
\begin{equation}
\label{eq: AvgDensityInequality}
    \overline{f^h(T)}\neq f^{\Bar{h}}(T),\qquad T\geq t_0,
\end{equation}
where $\overline{f^h(T)}$ is defined in \Cref{prop: ValueEPORIntegral}, and $\Bar{h}$ is defined -- with some abuse of notation -- as the function mapping $t$ into $\overline{h(t)}\equiv\E^\Q[h(t)|\F_h(t_0)]$, $t\geq t_0$.
We refer to the mismatch in the EPOR value resulting from substituting the right-hand side of \eqref{eq: AvgDensityInequality} in \eqref{eq: ValueEPORIntegral} as \emph{nonlinear adjustment}\footnote{In literature, it is more common to find the term \emph{convexity adjustment} and it always has the same sign. In our framework, however, this quantity may be positive or negative depending on the average housing market level.} and we show in the following section that such adjustment depends on the interplay between the (co)variance structure of the underlying stochastic HM and the Hessian operator of $f^h$ in $\Bar{h}$.

\begin{cor}[Nonlinear adjustment]
\label{cor: ConvexityAdjustment}
    Let us introduce the notation:
\begin{equation}
\begin{array}{cc}
    \begin{aligned}
         \Bar{h}:\:&\R^+\longrightarrow \R^+,\\
         & t\longmapsto\overline{h(t)},
     \end{aligned}
    \quad&\quad
    \begin{aligned}
         \Delta h:\:&\Omega\times\R^+\longrightarrow \R^+,\\
         &(\omega,t)\longmapsto h(\omega,t)-\Bar{h}(t),
     \end{aligned}
\end{array}
\end{equation}
    and let $V_{\Bar{h}}\equiv V_{\Bar{h}}(t_0)$ be defined as:
    \begin{equation*}
         V_{\Bar{h}}=\int_{t_0}^{T^*} C(T) f^{\Bar{h}}(T)\d T.
     \end{equation*}
    Then, pricing formula \eqref{eq: ValueEPORIntegral} in \Cref{prop: ValueEPORIntegral} allows for the representation:
    \begin{align}
    \label{eq: EPORValueApproximationCOnvexityAdj}
        V&= V_{\Bar{h}}+ \nu_{\Bar{h}} + \O\Big(\overline{||\Delta h||_{T^*}^3}\Big),\\
        \label{eq: NonlinearAdjustmentTerm}
        \nu_{\Bar{h}}&\equiv \nu_{\Bar{h}}(t_0) = \frac{1}{2}\int_{t_0}^{T^*}C(T)\overline{\langle \Delta h, \H^{\Bar{h}}_T\rangle_T}\d T.
    \end{align}
    where $\overline{||\Delta h||_{T}^3}\equiv\E^\Q\big[||\Delta h||_{T}^3\big|\F_h(t_0)\big]$ and $\overline{\langle \Delta h, \H^{\Bar{h}}_T\rangle_T}\equiv\E^\Q\big[\langle \Delta h, \H^{\Bar{h}}_T\rangle_T\big|\F_h(t_0)\big]$, with $||g_1||_T=\big(\int_{t_0}^T g_1(t)^2\d t\big)^{1/2}$ and $\langle g_1,g_2\rangle_T=\int_{t_0}^T\int_{t_0}^T g_2(s,t)g_1(s)g_1(t)\d s\d t$ for suitable, sufficiently integrable functions $g_{1}:\R^+\rightarrow \R$ and $g_2:\R^+\times\R^+\rightarrow \R$.
    For every fixed $T\in\T$, $\H^{\Bar{h}}_T$ is the Hessian of $f^h$ evaluated at the trajectory $\Bar{h}$, i.e.:
     \begin{equation*}
     \begin{aligned}
         \H^{\Bar{h}}_T:\:&\R^+\times\R^+\longrightarrow \R,\\
         &(s,t)\longmapsto \frac{\delta^2 f^h(T)}{\delta h(s) \delta h(t)}\bigg|_{h=\Bar{h}},
        \end{aligned}
    \end{equation*}
     where the functional derivative, $\frac{\delta^2 f^h(T)}{\delta h(s) \delta h(t)}$, is defined as:
     \begin{equation}
     \label{eq: HessianConvexityAdjustment}
     \begin{aligned}
         \frac{\delta^2 f^h(T)}{\delta h(s) \delta h(t)} &= f^h(T)\bigg(\Big(1-\Lambda(s,T)-\Lambda(t,T)\Big)\lambda'(h(s))\lambda'(h(t))\\
         &\qquad\qquad\qquad \quad+\Big(\Lambda(t,T) - 1\Big)\delta(s-t)\lambda''(h(s))\bigg),
     \end{aligned}
     \end{equation}
     with:
     \begin{equation*}
         \Lambda(v,T)=\frac{\delta(T-v)}{\lambda(h(T))},\qquad \lambda'\equiv\frac{\d\lambda}{\d h},\qquad \lambda''\equiv\frac{\d^2\lambda}{\d h^2},
     \end{equation*}
     and $\delta$ is the Dirac delta distribution, i.e. $\delta(x)=0$ for every $x\neq 0$ and $\int_{\R}\delta(x)\d x=1$.
\end{cor}
\begin{proof}
    The proof is given in \Cref{app: ProofConvexityAdjustment}.
\end{proof}
Intuitively, the nonlinear adjustment, $\nu_{\Bar{h}}$, acts as a correction term that captures most of the mismatch between $V$ -- the EPOR value resulting from a stochastic model for the HM activity -- and $V_{\Bar{h}}$ -- the value from a deterministic model based on the average HM. This is verified numerically in a simplified setting -- see \Cref{exm: FlatStochasticHM} below, for more details -- and shown in \Cref{fig: NonlinearAdj}a. The dotted-dashed blue line representing $V$ closely follows the dashed cyan line representing the sum of $V_{\Bar{h}}$ and $\nu_{\Bar{h}}$.  For reference, we represent the baseline value $V_{\Bar{h}}$ as a solid red line. The EPOR mismatch in the figure is plotted against the mortgage rate level, $K$, as a percentage of $V_{\Bar{h}}$, i.e. on the vertical axis we plot the quantity $(1-y/V_{\Bar{h}})\times 100$, for $y$ being either $V$, $V_{\Bar{h}}$ or $V_{\Bar{h}}+\nu_{\Bar{h}}$.

\begin{figure}[t]
    \centering
    \subfloat[\centering]{{\includegraphics[width=7.cm]{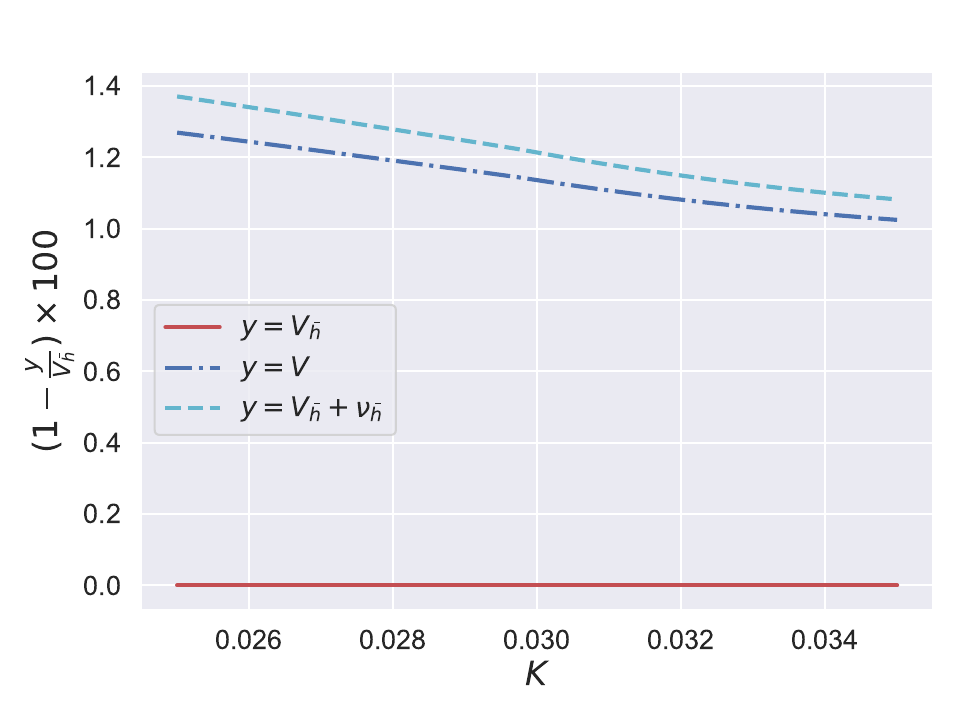} }}%
    ~\hspace{-.5cm}
    \subfloat[\centering]{{\includegraphics[width=7.cm]{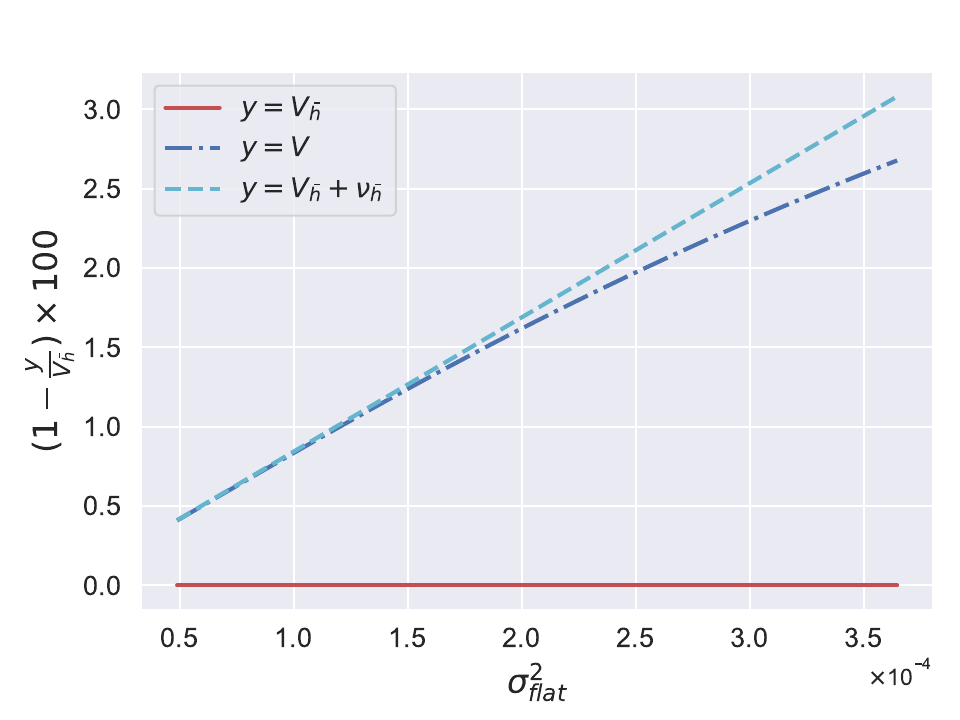} }}
    \caption{a) Relative impact of the stochastic model for different fixed rates. The dotted-dashed blue, the solid red, and the dashed cyan lines represent $(1-y/V_{\Bar{h}})\times 100$ for $y=V,V_{\Bar{h}},V_{\Bar{h}}+\nu_{\Bar{h}}$, respectively. b) Effect of variance on the nonlinear adjustment, in the framework described in \Cref{exm: FlatStochasticHM}.}
    \label{fig: NonlinearAdj}%
\end{figure}

The nonlinear adjustment, $\nu_{\Bar{h}}$, in \eqref{eq: NonlinearAdjustmentTerm} is driven by two main quantities, the (co)variance of the underlying factor $h$ and the Hessian $\H^{\Bar{h}}_T$, for $T\in\T$, evaluated at the average trajectory $\Bar{h}$, see \eqref{eq: HessianConvexityAdjustment} in \Cref{cor: ConvexityAdjustment}.
\begin{rem}[Interplay between (co)variance and Hessian]
Writing explicitly the argument of the expectation in \eqref{eq: NonlinearAdjustmentTerm}, we get:
\begin{align*}
    \langle \Delta h, \H^{\Bar{h}}_T\rangle_T&=\int_{t_0}^T\int_{t_0}^T \H^{\Bar{h}}_T(s,t) \Delta h(s)\Delta h(t)\d s \d t.
\end{align*}
Hence, by taking the expectation and applying Fubini's theorem, the nonlinear adjustment is a functional of the covariance of the housing market process $h$, weighted by the Hessian of the density.
In particular, we obtain the term:
\begin{equation}
\label{eq: ExpectationDoubleIntegral}
    \overline{\langle \Delta h, \H^{\Bar{h}}_T\rangle_T}=\int_{t_0}^T\int_{t_0}^T \H^{\Bar{h}}_T(s,t) \overline{\Delta h(s) \Delta h(t)} \d s\d t,
\end{equation}
with $\overline{\Delta h(s) \Delta h(t)}\equiv\E^\Q\big[\Delta h(s) \Delta h(t)\big|\F_h(t_0)\big]$ the covariance between $h(s)$ and $h(t)$.
\end{rem}
For the sake of intuition, in the following example, we assume the HM activity is stochastic at time $t_0$ and then kept flat over time, as illustrated in \Cref{fig: DensityDistributions}a.
\begin{exm}
\label{exm: FlatStochasticHM}
     Let us assume that, for scenario $\omega\in\Omega$, $h(\omega,t)=h_{flat}(\omega)$ for every $t\geq t_0$; $h_{flat}:\Omega \rightarrow \R^+$ is a suitable random variable describing the uncertainty of the future HM activity level at time $t_0$, and in particular $h_{flat}\in \F(t_0)$. Under this assumption, the covariance $\overline{\Delta h(s) \Delta h(t)}$ is equal to the variance $\sigma_{flat}^2$ of $h_{flat}$, for any $s,t$. Hence, \eqref{eq: ExpectationDoubleIntegral} reads:
     \begin{equation}
     \label{eq: ExpectationDoubleIntegralExample}
         \overline{\langle \Delta h, \H^{\Bar{h}}_T\rangle_T}=\sigma_{flat}^2\int_{t_0}^T\int_{t_0}^T \H^{\Bar{h}}_T(s,t) \d s\d t.
     \end{equation}
     The nonlinear adjustment \eqref{eq: NonlinearAdjustmentTerm} is, in this case, proportional to the variance, $\sigma_{flat}^2$, of the HM random variable $h_{flat}$. We show such an effect in \Cref{fig: NonlinearAdj}b, where the nonlinear adjustment is depicted as a straight dashed cyan line. Furthermore, given that $C(T)$ in \eqref{eq: NonlinearAdjustmentTerm} is nonnegative, we notice that the Hessian, $\H^{\Bar{h}}_T$, determines the sign of the correction $\nu_{\Bar{h}}$ of \Cref{cor: ConvexityAdjustment}. 
\end{exm}

\begin{figure}[b]
    \centering
    \subfloat[\centering]{{\includegraphics[width=7.cm]{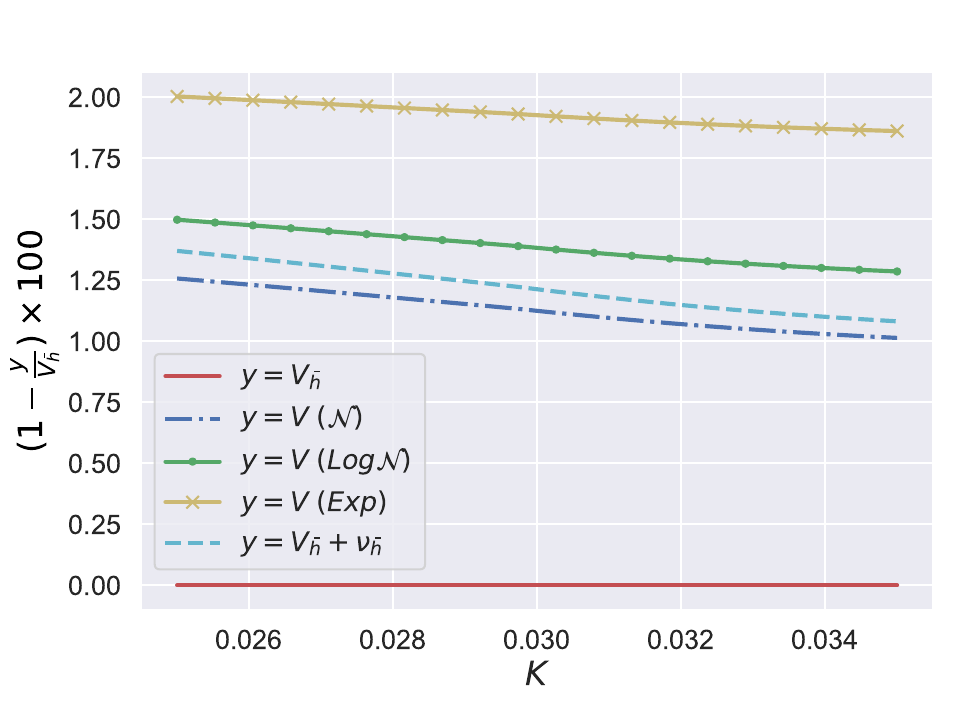} }}%
    ~\hspace{-.5cm}
    \subfloat[\centering]{{\includegraphics[width=7.cm]{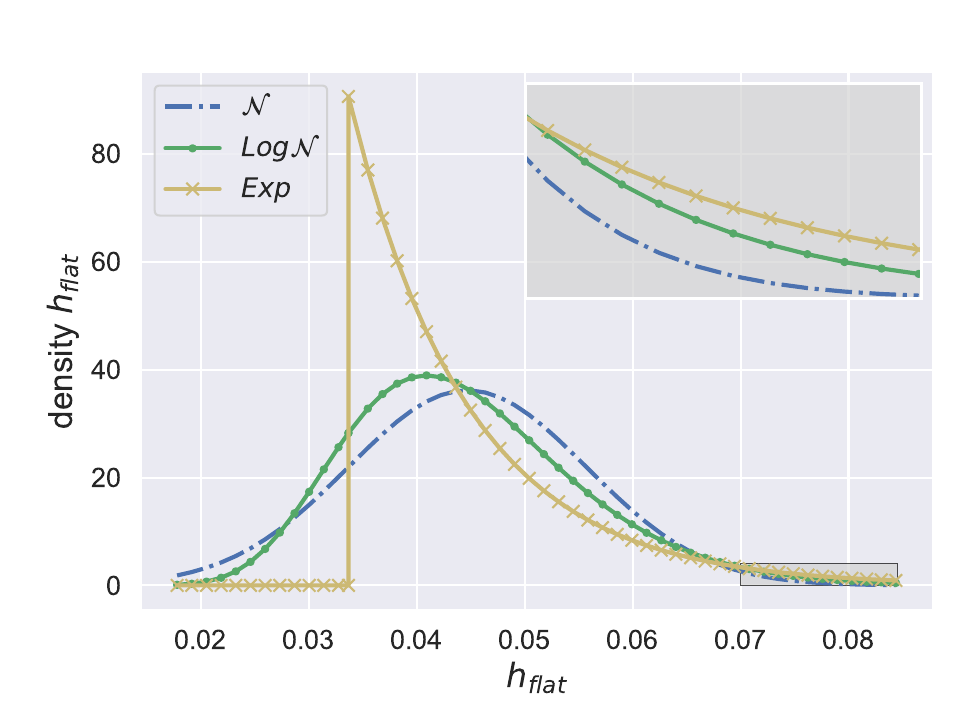} }}
    \caption{a) $V_{\Bar{h}}-V$ for normal ($\N$), log-normal ($log\N$) and exponential ($Exp$) $h_{flat}$ compared with (minus) the non-linear adjustment $\nu_{\Bar{h}}$. b) Densities of the three distributions above with a focus on their right tails.}
    \label{fig: NonlinearAdjDifferentDensities}%
\end{figure}

\begin{rem}[Accuracy deterioration of the non-linear adjustment]
    In the setting of \Cref{exm: FlatStochasticHM}, considering a normal random variable for $h_{flat}$, we check the accuracy of the nonlinear adjustment in capturing the mismatch between the deterministic and the stochastic price deteriorates for a stochastic model with higher variance (see the dash-dotted blue line in \Cref{fig: NonlinearAdj}b). This effect is a consequence of the Taylor expansion-based price approximation in \Cref{cor: ConvexityAdjustment} being less accurate when $h$ deviates from its average. In general, the deterioration is more pronounced when we consider a random variable $h_{flat}$ with more mass in the distribution tails (hence, a more pronounced $O(\overline{||\Delta h||_{T^*}^3})$ in \eqref{eq: EPORValueApproximationCOnvexityAdj}). In \Cref{fig: NonlinearAdjDifferentDensities}, we illustrate this fact using as examples the normal, log-normal, and (shifted) exponential distributions. \Cref{fig: NonlinearAdjDifferentDensities}a shows that the nonlinear adjustment (dashed cyan line) captures the effect of a normal $h_{flat}$ (dash-dotted blue line), while it performs worse when $h_{flat}$ is log-normal (solid-dotted green line) or exponential (solid-crossed yellow line). In \Cref{fig: NonlinearAdjDifferentDensities}b, we show the densities of the three distributions, focusing on the mass in their right tails.
\end{rem}

Because of the Dirac deltas appearing in \eqref{eq: HessianConvexityAdjustment}, for any $T\in\T$, the Hessian $\H_T^{\Bar{h}}$ is concentrated mainly in $\H_T^{\Bar{h}}(T,T)$, with a second order effect in $\H_T^{\Bar{h}}(T,t)$, $t\in[t_0,T)$, and $\H_T^{\Bar{h}}(s,T)$, $s\in[t_0,T)$. To illustrate such an effect, in \Cref{fig: HessianHeatmapAndIntegral}a, we report a discretized version of the Hessian for $T=10$ and assuming $h$ is interpolated linearly between a set of predefined spine points. The discretization details are given in \Cref{app: DiscretizedHessian}. In \Cref{fig: HessianHeatmapAndIntegral}b, the double integral in \eqref{eq: ExpectationDoubleIntegralExample} is depicted as a function of $T\in\T$, i.e. $T\mapsto \int_{t_0}^T\int_{t_0}^T \H^{\Bar{h}}_T(s,t) \d s\d t$. In \Cref{exm: FlatStochasticHM}, the sign of the double integral determines the sign of the non-linear adjustment, being both $\sigma^2_{flat}$ in \eqref{eq: ExpectationDoubleIntegralExample} and $C(T)$ in \eqref{eq: NonlinearAdjustmentTerm} strictly positive.

\begin{figure}[t]
    \centering
    \subfloat[\centering]{{\includegraphics[width=7.cm]{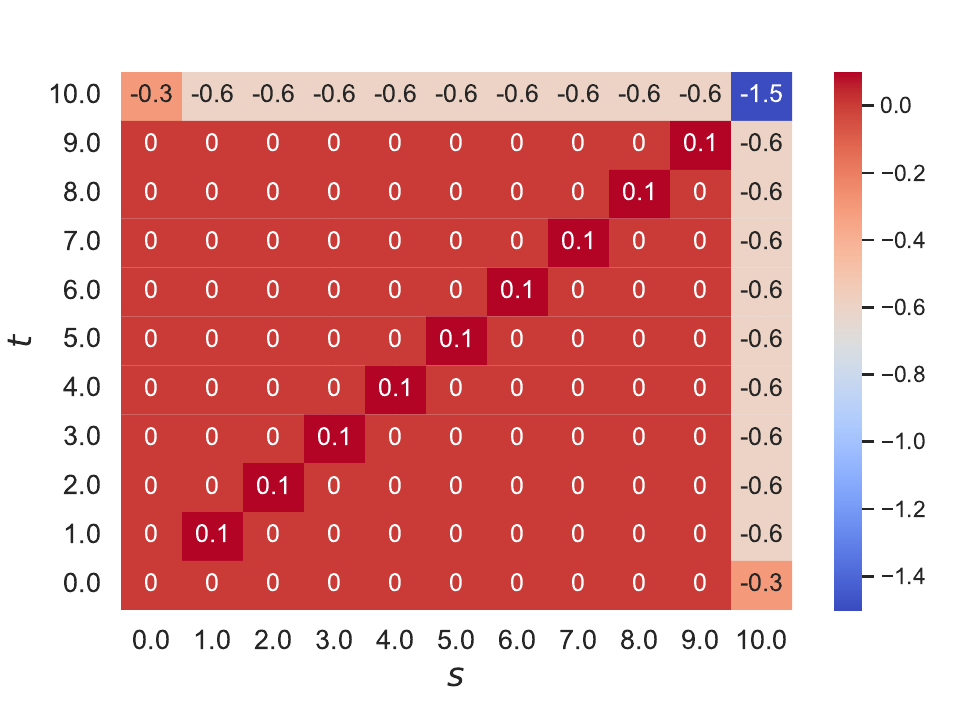} }}%
    ~\hspace{-.5cm}
    \subfloat[\centering]{{\includegraphics[width=7.cm]{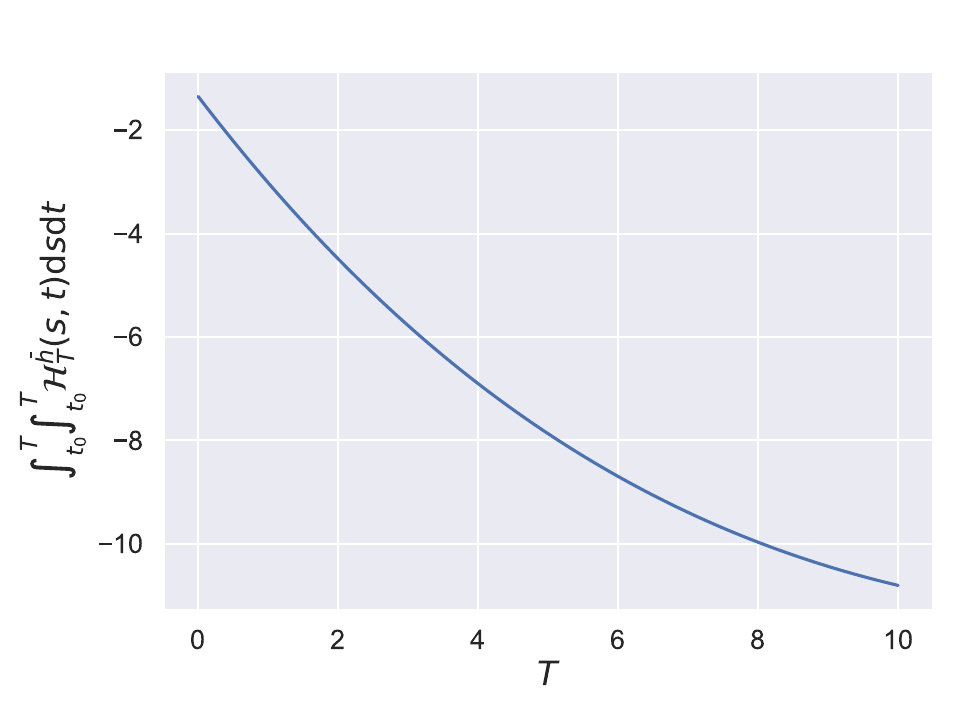} }}
    \caption{Left: discretized Hessian function. Right: double integral $\int_{t_0}^T\int_{t_0}^T \H^{\Bar{h}}_T(s,t) \d s\d t$ of the Hessian function for different $T$.}
    \label{fig: HessianHeatmapAndIntegral}%
\end{figure}

\section{Hedging of the EPOR}
\label{sec: Hedging}

The hedging of the EPOR is performed in a standard setting using market tradable instruments such as IRSs and interest rate swaptions. The underlying insight is that by investing in tradable instruments replicating the exposure first derivative w.r.t. the risk factors, the total position value is invariant to small changes in the risk factors. In the financial jargon, this technique is referred to as \emph{Delta hedging}, since it aims to replicate -- with the opposite sign -- the ``Delta'' of the exposure using market instruments. Assuming continuous monitoring and rebalancing of the hedging portfolio, simple Delta hedging guarantees ``perfect'' protection regarding all market risk factors. 

Fixing a set of $I$ market instrument prices, $\S_{yc}=[\S_{yc,1},\dots,\S_{yc,I}]^\top$, used to calibrate the prevailing yield curve (YC), the Delta of the EPOR is the gradient of the EPOR price w.r.t. $\S_{yc,i}$, for $i=1,\dots,I$. Calling $\S=[\S_1,\dots,\S_J]^\top$ the prices of the $J$ market instruments used for hedging and $w=[w_1,\dots,w_J]^\top$ the corresponding weights (i.e. the notional), the optimal Delta hedge is achieved by solving the linear system in $w$:
\begin{equation}
\label{eq: DeltaSystem}
    \Delta(V) = \sum_{j=1}^J w_{j} \Delta(\S_{j}),
\end{equation}
where the operator $\Delta(\cdot)$ is the gradient in $\S_{yc}$ of the instrument price given as argument. Particularly, $\Delta(V)\in\R^I$ and $\Delta(\S_{j})\in\R^{I}$ read:
\begin{equation*}
    \Delta(V)=\begin{bmatrix}
        \frac{\partial V}{\partial \S_{yc,1}}\\
        \vdots\\
        \frac{\partial V}{\partial \S_{yc,I}}
    \end{bmatrix},\qquad
    \Delta(S_{j})=\begin{bmatrix}
        \frac{\partial \S_{j}}{\partial \S_{yc,1}}\\
        \vdots\\
        \frac{\partial \S_{j}}{\partial \S_{yc,I}}
    \end{bmatrix},
    \quad j=1,\dots,J.
\end{equation*}

In a more realistic framework, hedging should be parsimonious. Therefore, more robust strategies allowing for less frequent rebalancing dates are preferred. This is achieved by replicating the second derivatives of the exposure w.r.t. the market risk factors. In the financial jargon, \emph{Delta-Gamma hedging} aims to replicate the exposure ``Gamma,'' in addition to the Delta. Delta-Gamma hedging is formalized by including additional conditions based on the second derivatives, leading to the system:
\begin{equation}
\begin{aligned}
\label{eq: DeltaGammaSystem}
    \Delta(V) &= \sum_{j=1}^J w_{j} \Delta(\S_{j}),\\
    \Gamma(V) &= \sum_{j=1}^J w_{j} \Gamma(\S_{j}),
\end{aligned}
\end{equation}
where the operator $\Gamma(\cdot)$ is the Hessian in $\S_{yc}$ of the instrument price given as argument.
Particularly, $\Gamma(V)\in\R^{I\times I}$ and $\Gamma(\S_{j})\in\R^{I\times I}$, $j=1,\dots,J$, read:
\begin{equation*}
    \Gamma(V)=\begin{bmatrix}
            \frac{\partial^2 V}{\partial \S_{yc,1}^2}&\cdots& \frac{\partial^2 V}{\partial \S_{yc,1}\partial \S_{yc,I}}\\
            \vdots&\ddots&\vdots\\
            \frac{\partial^2 V}{\partial \S_{yc,I}\partial \S_{yc,1}}&\cdots& \frac{\partial^2 V}{\partial \S_{yc,I}^2}
        \end{bmatrix},\qquad
        \Gamma(S_j)=\begin{bmatrix}
            \frac{\partial^2 \S_j}{\partial \S_{yc,1}^2}&\cdots& \frac{\partial^2 \S_j}{\partial \S_{yc,1}\partial \S_{yc,I}}\\
            \vdots&\ddots&\vdots\\
            \frac{\partial^2 \S_j}{\partial \S_{yc,I}\partial \S_{yc,1}}&\cdots& \frac{\partial^2 \S_j}{\partial \S_{yc,I}^2}
            \end{bmatrix}.
\end{equation*}

The choice of the hedging instruments is crucial since it determines the existence (and uniqueness) of a solution of \eqref{eq: DeltaSystem} and \eqref{eq: DeltaGammaSystem}.
In general, a standard Delta-hedge is based on linear tradable instruments with (close to\footnote{Interest rate products, such as swaps, show non-zero Gamma because of the dependence of the discount factor on the stochastic interest rate itself (see \Cref{app: SwapDeltaGamma}). However, the Gamma of a swap is orders of magnitude smaller than the Gamma of a swaption written on the same underlying swap.}) vanishing Gamma. Instruments bearing non-vanishing Gamma are necessary for Delta-Gamma hedging, for \eqref{eq: DeltaGammaSystem} is trivially inconsistent otherwise.

In practice, hedging is typically \emph{parsimonious}, meaning that only a few instruments are used to replicate the total exposure. This often leads to overdetermined systems in \eqref{eq: DeltaSystem} and \eqref{eq: DeltaGammaSystem}, generally not allowing for a solution. Hence, the problem of hedging is recast as a minimization problem:
\begin{equation}
\label{eq: MinimizationProblemHedge}
    \min_{\w} ||\zeta_\Delta(\w)||_{\Delta}^2 + k ||\zeta_\Gamma(\w)||_{\Gamma}^2,
\end{equation}
where $\zeta_\Delta(\w) =  \sum_{j=1}^J w_{j} \Delta(\S_{j}) - \Delta(V)\in\R^I$ and $\zeta_\Gamma(\w) =  \sum_{j=1}^J w_{j} \Gamma(\S_{j}) - \Gamma(V)\in\R^{I\times I}$ are the Delta and Gamma mismatches between the hedge and the EPOR, $||\cdot||_\Delta$ and $||\cdot||_\Gamma$ are suitable norms, and $k\geq 0$ is a constant that allows stirring the focus from the Delta to the Gamma (with $k=0$ corresponding to \eqref{eq: DeltaSystem}). When $||\cdot||_\Delta$ is the Euclidean norm and $||\cdot||_\Gamma$ is the Frobenius norm, the objective function is quadratic in $\w$ and its optimal solution is analytically available.

\begin{rem}[Lack of interpretability]
    It is worth pointing out that an explainable hedge is desirable. For instance, if a shock is observed on a specific region of the YC, we would like to be able to adjust the hedge only \emph{locally}, i.e. only rebalancing the investment in the instruments ``responsible'' for the shocked region without completely changing our hedge. However, a hedge computed by solving \eqref{eq: MinimizationProblemHedge} lacks this feature since the weights are computed \emph{globally}. 
    \end{rem}
In the next section, we elaborate on this \emph{interpretability} issue and propose a solution.

\subsection{Explainable hedge}
\label{ssec: ExplainableHedge}

Because of the linear EPOR pricing formula structure in \eqref{eq: ValueEPORIntegral}, it is clear that an appropriate set of swaptions $C(T_j)\equiv C(T_j;K;\T_{\tt{p}},N)$, $j=1,\dots,J$, as defined in \Cref{prop: ValueEPORIntegral}, should provide an accurate hedge for the EPOR exposure. We expect that a suitable linear combination of receiver swaptions closely resembles the Delta and the Gamma of the EPOR. In the notation of the previous section, the hedging instruments are selected as $\S_j=C(T_j)$, where $T_j$ is the maturity of a swaption written on a swap with a notional profile corresponding to the outstanding mortgage notional and fixed rate equal to $K$. 

We divide the integration range $\T=[t_0,T^*]$ in \eqref{eq: ValueEPORIntegral} into subintervals $\RR_j=[R_{j-1},R_{j}]$ such that $\T=\cup_j \RR_j$. By linearity of the integration operator, the EPOR value can be written as:
\begin{equation*}
    V = \sum_j V_j,\qquad V_j\equiv V_j(t_0)=\int_{R_{j-1}}^{R_{j}} C(T) \overline{f^{h}(T)} \d T,
\end{equation*}
i.e. every $V_j$ is the component of the EPOR value conditional to relocation happening in the subinterval $\RR_j$.
To enhance the explainability of the hedge, we solve $J$ ``local'' minimization problems similar to the problem in \eqref{eq: MinimizationProblemHedge}. Particularly, the local minimization problems read:
\begin{equation}
\label{eq: LocalMinimizationProblem}
    \min_{w_j} ||\zeta_{\Delta,j}(w_j)||_\Delta^2 + k_j ||\zeta_{\Gamma,j}(w_j)||_\Gamma^2,\qquad j=1,\dots,J,
\end{equation}
where $\zeta_{\Delta,j}(w_j) = w_j \Delta(\S_j) - \Delta(V_j)$, $\zeta_{\Gamma,j}(w_j) = w_j\Gamma(\S_j) - \Gamma(V_j)$, and $k_j$ are nonnegative constants.
By solving the local minimization problems, we end up with a hedging strategy that is generally suboptimal for the global minimization problem in \eqref{eq: MinimizationProblemHedge}. However, the gain in this setting is an intuitive hedging strategy. The weights of the hedging strategy closely relate to the density of the underlying relocation probability, allowing for quick adjustment of the hedge when movements in the underlying housing market are expected. 

When the subintervals, $\RRR=[\RR_1,\dots,\RR_J]$, and the hedging swaption maturities, $\TTT=[T_1,\dots,T_J]$, are fixed, we indicate as $w^*_j\equiv w^*_j(\RR_j,T_j)$ the optimal weights solving the local minimization problems in \eqref{eq: LocalMinimizationProblem} and $w^*\equiv w^*(\RRR,\TTT)$ the vector with components $w_j^*$. It is not surprising that $w^*_j$ depends on $\RR_j$ and $T_j$. In fact, $\RR_j$ affects $\Delta(V_j)$ and $\Gamma(V_j)$, while $T_j$ has an effect on $\Delta(S_j)$ and $\Gamma(S_j)$. As a consequence, even if the weights $w_j^*$ are optimal for \eqref{eq: LocalMinimizationProblem}, we might be interested in improving the hedge by \emph{optimally selecting $\RR_j$ or $T_j$}.
For instance, given a fixed subinterval selection $\RR_j$, $j=1,\dots,J$, we find the optimal maturity $T_j^*$, within each subinterval, by solving:
\begin{equation}
\label{eq: LocalMinimizationProblemOptimalMaturity}
    \min_{T_j\in\RR_j} ||\zeta_{\Delta,j}(w^*_j(\RR_j,T_j))||_\Delta^2 + k_j ||\zeta_{\Gamma,j}(w^*_j(\RR_j,T_j))||_\Gamma^2,\qquad j=1,\dots,J,
\end{equation}
ensuring that an \emph{optimal swaption maturity $T_j^*$} is selected to replicate the Greeks in every subinterval $\RR_j$.

Alternatively, we fix the criterion to select each maturity given a certain subinterval, i.e. $T_j\equiv T_j(\RR_j)$ for $j=1.\dots,J$, and in general $\TTT=\TTT(\RRR)$, and choose the optimal subintervals $\RRR^*$ to minimize the total mismatch in the Greeks. Formally, this is enforced by solving:
\begin{equation}
\label{eq: LocalGlobalMinimizationProblemOptimalSubintervals}
    \min_{\RRR} ||\zeta_\Delta(w^*(\RRR)||_{\Delta}^2 + k ||\zeta_\Gamma(w^*(\RRR)||_{\Gamma}^2 + k_{vol}Vol(\RRR),
\end{equation}
where the first two terms are taken from the global minimization problem \eqref{eq: MinimizationProblemHedge} and the last term improves the stability of the solution. The weights $w^*(\RRR)\equiv w^*(\RRR,\TTT(\RRR))$ are optimal solutions of the local minimization problems \eqref{eq: LocalMinimizationProblem} for a given subinterval selection $\RRR$ and maturities $\TTT(\RRR)$. Problem \eqref{eq: LocalGlobalMinimizationProblemOptimalSubintervals} aims to compute the subinterval selection, $\RRR^*$, that minimizes the total mismatch in the Greeks in the sense of \eqref{eq: MinimizationProblemHedge} while maintaining the explainability of \eqref{eq: LocalMinimizationProblem} (i.e. only one instrument is used to replicate the EPOR Greeks in every subinterval). The term $Vol(\RRR)$ is defined as:
\begin{equation*}
    Vol(\RRR)= \Bigg(1 - \frac{\prod_{j=1}^J\ell_j}{\Bar{\ell}^J}\Bigg)^J,
\end{equation*}
with $\ell_j$ the length of the subinterval $\RR_j$ and $\Bar{\ell}=\frac{1}{J}\sum_j\ell_j$. Geometrically, the term $Vol(\RRR)$ represents the mismatch in hypervolume between the $J$-dimensional hyperrectangle with side lengths $\ell_j$, $j=1,\dots,J$, and the corresponding hypercube with the same hyperperimeter, i.e. $\sum_j \ell_j$. For $k_{vol}>0$, \eqref{eq: LocalGlobalMinimizationProblemOptimalSubintervals} prevents the subintervals from close to zero lengths while prioritizing a balanced partition of the integration interval.

\subsection{Actuarial hedge}
\label{ssec: ActuarialHedge}

From \Cref{prop: ValueEPORIntegral}, we infer that an \emph{implied HM} trajectory exists such that it yields a relocation density equal to its average. 
This means that the classic Delta-Gamma hedge described in the previous sections is optimal only when the implied HM is observed, and it might perform poorly when the realized HM is different.

As observed in \Cref{cor: EPORDIstribution}, the EPOR value is the expectation over all the HM scenarios of the EPOR value, conditional to a given HM realization. Hence, we may want to build a hedging strategy that protects us against a more general HM outcome. To do so, we rely on the geometric interpretation of the Gamma profile. The Gamma matrix encapsulates information regarding the curvature of the exposure with respect to the quotes of the instruments used to calibrate the yield curve. Particularly, we aim to generate a hedging strategy that matches the exposure Delta as accurately as possible while ensuring a convex total exposure, i.e. we look for a hedging strategy that ensures an increase in value no matter what movements are observed on the calibrating market quotes. Specifically, we add a penalty in the objective of the optimization based on the sign of the eigenvalues of the Hessian matrix for every subinterval $\RR_j$. Ideally, we look for hedging instruments such that the Hessian is positive semidefinite. However, this is in general difficult to achieve, but we can impose the Hessian to be ``as close as possible'' to a positive semidefinite matrix. By ``as close as possible'' we mean that even if an eigenvalue remains negative, its magnitude is reduced as much as possible. Geometrically, along the directions where the total exposure is concave, we reduce the curvature so that a movement along that direction -- the associated eigenvector -- entails a smaller decrease in value. 

Starting from \eqref{eq: LocalMinimizationProblem}, the hedging problem in every subinterval is updated as:
\begin{equation}
\label{eq: MinimizationEigen}
     \min_{w_j} ||\zeta_{\Delta,j}(w_j)||_\Delta^2 + k_j ||\zeta_{\Gamma,j}(w_j)||_\Gamma^2 -k_{eig} \ES_\alpha[\min(\Lambda(\zeta^h_\Gamma(w_j)))],\qquad j=1,\dots,J,
\end{equation}
where $\zeta_{\Delta,j}(w_j)$, $\zeta_{\Gamma,j}(w_j)$, and $k_j$ are as in \eqref{eq: LocalMinimizationProblem}, and $\zeta_{\Gamma,j}^h(w_j)=w_{j} \Gamma(S_{j}) - \Gamma(V_{h,j})$ is the Gamma mismatch in the subinterval $\RR_j$ given HM $h$, i.e. computed using $V_{h,j}=\int_{R_{j-1}}^{R_j} C(T)f^h(T)\d T$. $\ES_\alpha$ is the \emph{expected shortfall} at the $alpha\in[0,1]$ level, i.e. $\ES_{\alpha}[X]:=\E[X|X\leq Q_{\alpha}(X)]$ for $Q_{\alpha}(X)$ the $\alpha$-quantile of the random variable X. $\Lambda(\cdot)$ represents the eigenvalues of the matrix given as argument and $\min(\cdot)$ selects the minimum eigenvalue. $k_{eig}\geq 0$ is a weight term. The first two terms in \eqref{eq: MinimizationEigen} aim to generate a hedging strategy that closely approximates the Delta and the Gamma of the EPOR for each subinterval $\RR_j$. In particular, this is equivalent to matching the average Delta and Gamma profiles $V_{h,j}$. The last term penalizes the concavity along the most significant direction, for the worst possible realization of $h$, i.e. the ``worst'' $\alpha$. Based on the magnitude of $k_{eig}$, the minimization problem generates a hedging strategy that might show a significant mismatch in Delta, compared to the EPOR Delta. For this reason, after the minimization, we improve the strategy by adding a set of IRSs to correct the Delta profile. This does not significantly affect the concavity of the total position since IRSs bear close to zero Gamma. As a result, we obtain a combination of market instruments -- swaps and swaptions -- that closely matches the Delta of the EPOR exposure, while controlling the concavity of the total position.

\section{Data and calibration}
\label{sec: DataCalibration}

In this section, we present the methodology used to calibrate the relocation density. We used time series data, with monthly frequency, starting in December 2012 and ending in November 2023. For each month, we observe the fraction of transactions on a national level and the fraction of borrowers who decided to relocate to a new house. Recalling the notation of \eqref{eq: HMActivity}, the fraction of transactions on the national level observed in a given month starting at time $t$ is given by $h_{\Delta t}(t)=NoT(t,t+\Delta t)/NoH(t)$, $\Delta t = 1/12$. As a proxy for $h(t)$ in \eqref{eq: HMActivity} we use its \emph{discrete} analogous, i.e. $h(t)\approx h_{\Delta t}(t)/\Delta t$. Based on $h(t)$, we perform two distinct calibrations. On the one hand, we assume each $h(t)$ is a realization of a possible HM activity and we use standard moment estimators for expected value and variance to find distributions matching the data. We calibrate normal, lognormal, and (shifted) exponential distributions, respectively. The empirical mean and variance observed on the data are $\hat{\mu}_h=4.470\times 10^{-2}$ and $\hat{\sigma}_h^2=1.215\times 10^{-4}$. The normal, lognormal, and (shifted) exponential distributions for $h$ are used to present the effects of a stochastic HM activity on the EPOR value (see e.g. \Cref{fig: NonlinearAdj,fig: NonlinearAdjDifferentDensities}). In \Cref{fig: MLEHMtoRelocProb}a, the empirical distribution of $h$ is shown (cyan histogram) with the densities of the calibrated normal (dash-dotted blue line) and log-normal (solid-dotted blue line) random variables.  On the other hand, using MLE on the time series of $h$ (see \Cref{fig: MLEHMtoRelocProb}b), we calibrate an Ornstein-Uhlenbeck (OU) process, i.e.:
\begin{equation*}
    \d h(t) = \alpha_h(\theta_h(t) - h(t))\d t + \eta_h\d W^{\P}(t),
\end{equation*}
where $\alpha_h$ is the mean-reversion rate, $\theta_h(t)$ is the time-dependent long-term mean, and $\eta_h$ is the volatility. $\P$ represents the real-world measure. Such a process is used as a scenario generator for future HM activity. The mean-reversion rate and volatility coefficients are $\alpha_h=126$ and $\eta_h=0.115$, respectively. We assume a time-dependent mean that is used to capture possible trends. When used for scenario generation, a suitable time-dependent mean allows the user to include in the model his/her expectation regarding future HM activity. For instance, in the numerical experiments, we show examples where different HM ``trends'' are considered, flat (i.e., $\theta(t)=\hat{\mu}_h$), increasing (e.g., $\theta(t)=\hat{\mu}_h+2\hat{\sigma}_h\frac{t-t_0}{T^*-t_0}$) and decreasing (e.g., $\theta(t)=\hat{\mu}_h-2\hat{\sigma}_h\frac{t-t_0}{T^*-t_0}$).

\begin{figure}[t]
    \centering
    \subfloat[\centering]{{\includegraphics[width=6.9cm]{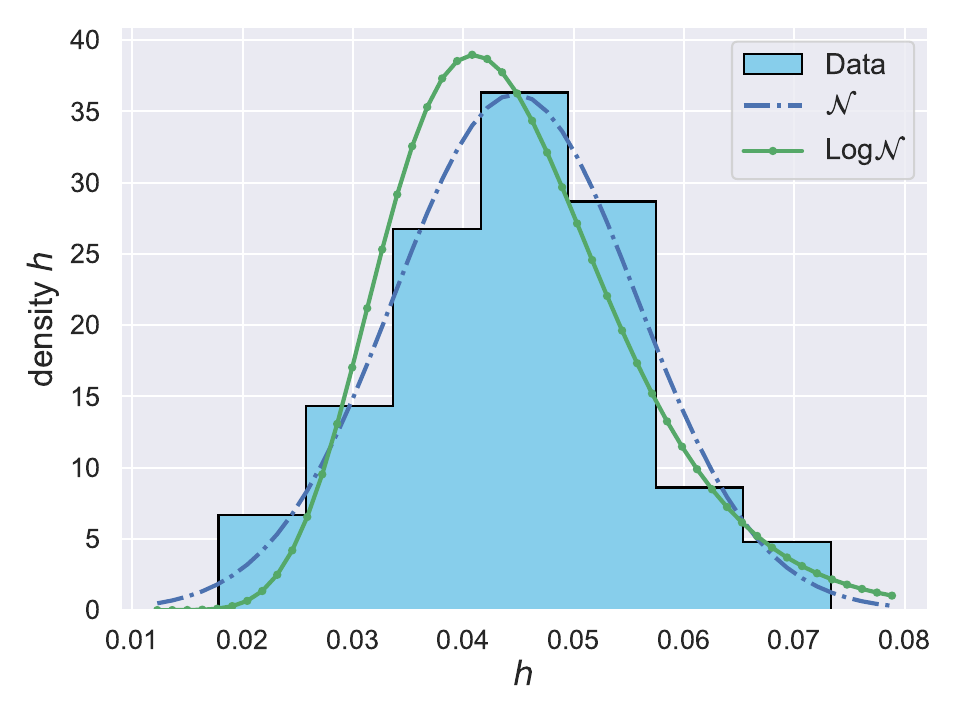} }}%
    ~\hspace{-.5cm}
    \subfloat[\centering]{{\includegraphics[width=7.5cm]{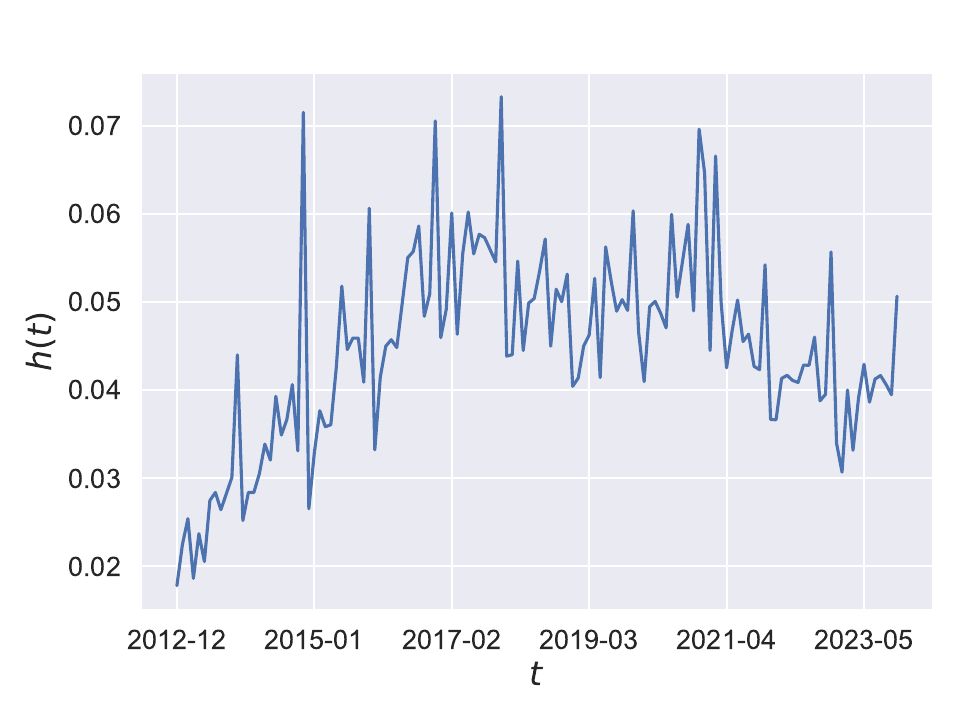} }}
    \caption{\footnotesize Left: Empirical distribution of the observed HM activity (cyan histogram), calibrated normal and lognormal densities (dashed-dotted blue and dotted green lines, respectively). Right: HM activity time series.}
    \label{fig: MLEHMtoRelocProb}%
\end{figure}

The same data are used to infer the relationship between the relocation density intensity $\lambda(t)$ and the HM activity $h(t)$. The probability of a borrower's relocation in a certain period, given no relocation has occurred before, is defined as $p_{\Delta t}(t)=\P\big[t < \tau \leq t + \Delta t\big| \tau > t\big]$ and is approximated with the fraction of borrowers relocating in the observed month. We assume a logistic function with polynomial argument maps $h_{\Delta t}(t)$ into $p_{\Delta t}(t)$. In particular, we set the equation:
\begin{equation*}
    p_{\Delta t}(t|\beta_{\Delta t})=\frac{1}{1+\e^{-\beta_{\Delta t}^\top H_{\Delta t}(t)}},
\end{equation*}
where $H_{\Delta t}(t)=[1,h_{\Delta t}(t),h_{\Delta t}^2(t)]^\top$ and $\beta_{\Delta t}$ are the coefficients of the polynomial $\beta_{\Delta t}^\top H_{\Delta t}(t)$. The optimal coefficients $\beta_{\Delta t}^*$ are obtained by maximizing the likelihood function conditional to the available data.
Then, the intensity $\lambda(t)$ is represented as a function of the HM activity $h(t)$ using the relationship $p_{\Delta t}(t) \approx \lambda(t)\Delta t$ for \emph{small} $\Delta t$. We end up with the expression for the intensity:
\begin{equation}
\label{eq: MLEEstimatorLambda}
    \lambda(t|\beta^*)\approx\frac{p_{\Delta t}(t|\beta_{\Delta t}^*)}{\Delta t}=\frac{\Delta t^{-1}}{1+\e^{-{\beta^*}^\top H(t)}},
\end{equation}
where $H(t)=[1,h(t), h(t)^2]$ and the optimal coefficients $\beta^*$ are scaled accordingly with the time discretization, i.e. $\beta^*_i=\beta_{\Delta t,i}^* \Delta t^i$, for $i=0,1,2$. The calibrated coefficients equal $\beta_0^*=-7.50$, $\beta_1^*=54.18$, $\beta_2^*=-326.86$.
In \Cref{fig: LambdaTimeseries}a, we report the calibrated intensity $\lambda$ as a function of the HM activity $h$, whereas \Cref{fig: LambdaTimeseries}b shows that the observed intensity $p_{\Delta t}(t)/\Delta t$ (solid red line) is closely approximated by $\lambda(t|\beta^*)$ (dash-dotted green line) computed using \eqref{eq: MLEEstimatorLambda}.

\begin{figure}[t]
    \centering
    \subfloat[\centering]{{\includegraphics[width=7.2cm]{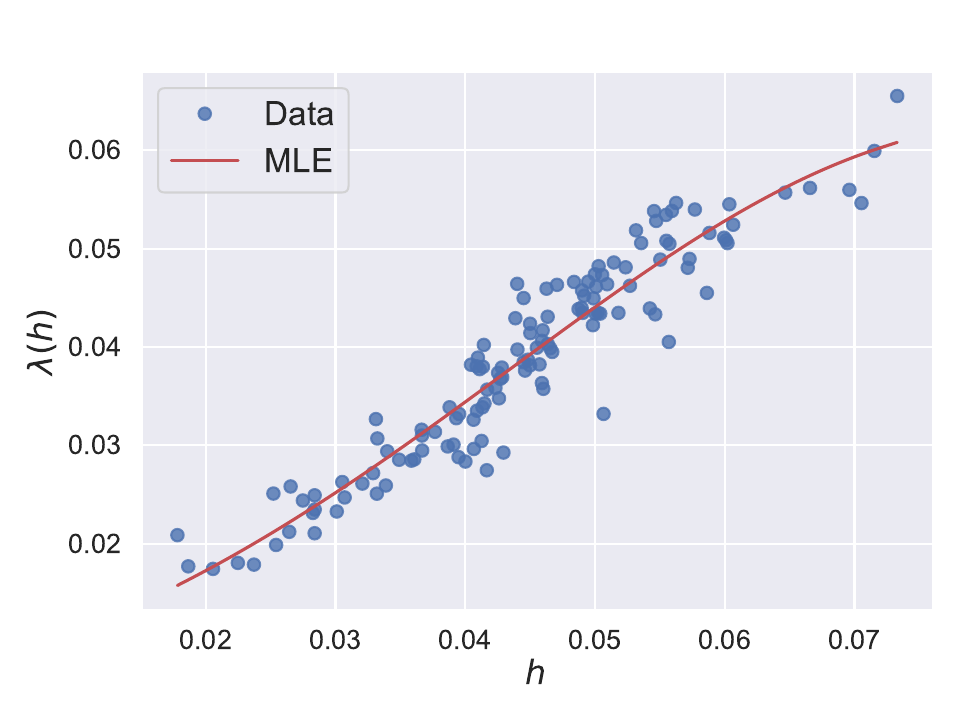} }}%
    ~\hspace{-.5cm}
    \subfloat[\centering]{{\includegraphics[width=7.2cm]{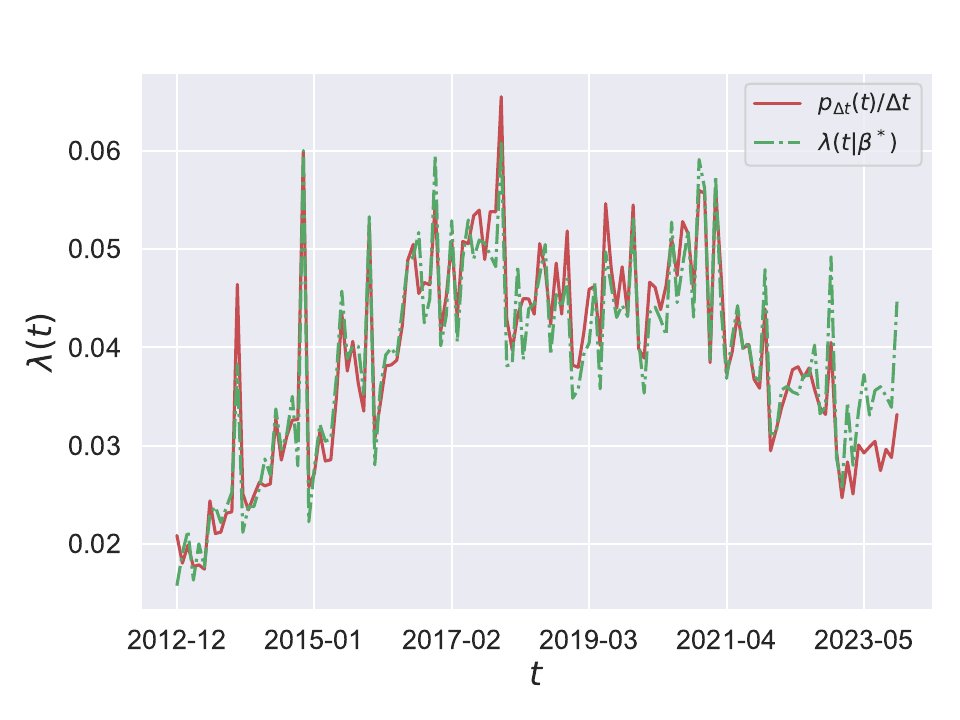} }}
    \caption{\footnotesize Left: Mapping from $h$ to $\lambda(h)$. Data vs MLE. Right: Comparison between observed and approximated intensity.}
    \label{fig: LambdaTimeseries}%
\end{figure}

\section{Experiments}
\label{sec: Experiments}

\begin{figure}[b]
    \centering
    \subfloat[\centering]{{\includegraphics[width=7.1cm]{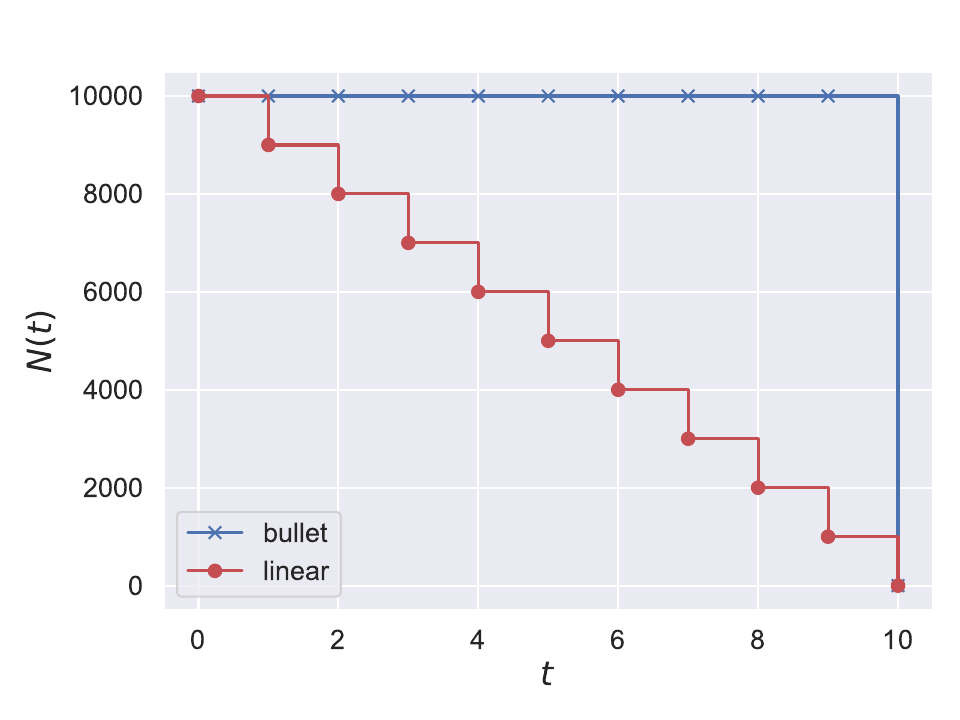} }}%
    ~\hspace{-.2cm}
    \subfloat[\centering]{{\includegraphics[width=7.1cm]{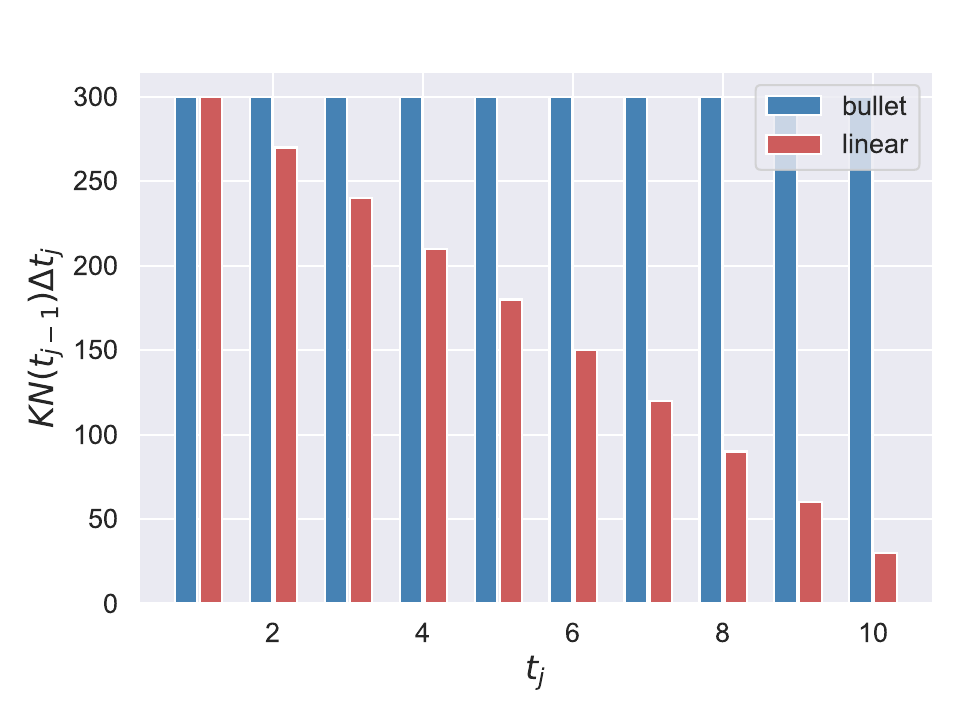} }}
    \caption{\footnotesize Left: Amortization schemes for bullet (crossed blue line) and linear (dotted red line). Right: corresponding interest payments.}
    \label{fig: NotionalProfiles}%
\end{figure}

\begin{figure}[b!]
    \centering
    \subfloat[\centering]{{\includegraphics[width=7.1cm]{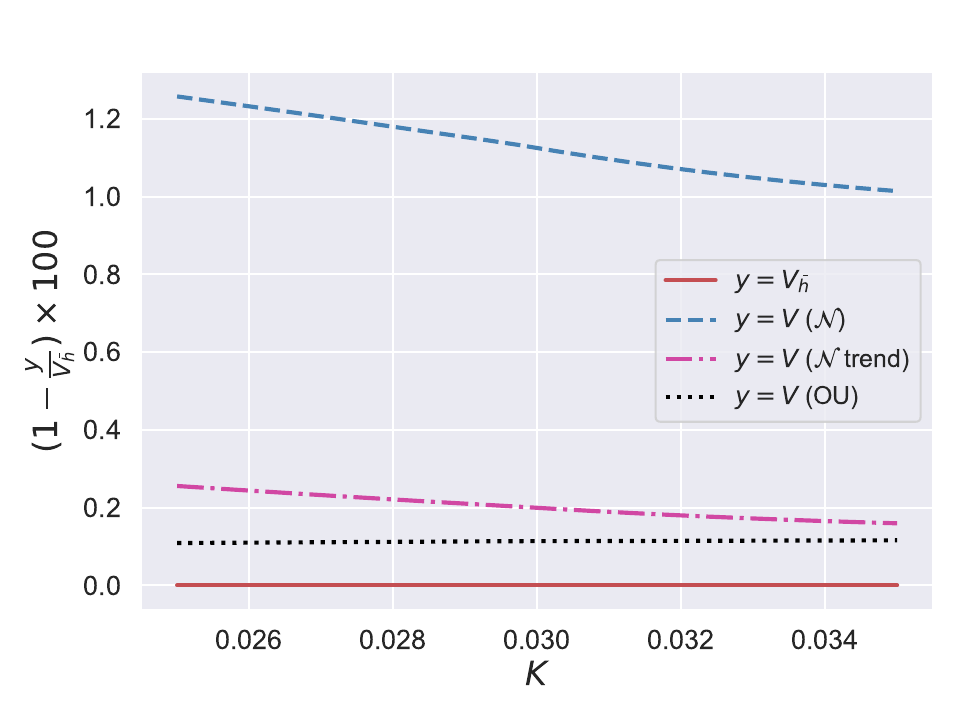} }}%
    ~\hspace{-.2cm}
    \subfloat[\centering]{{\includegraphics[width=7.1cm]{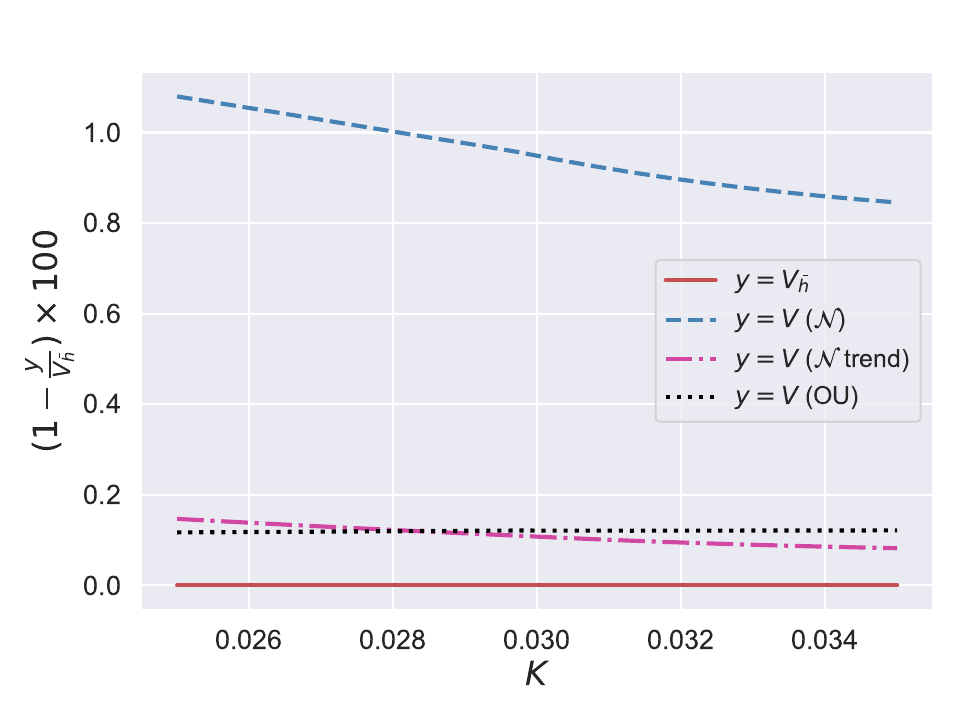} }}\\
    \subfloat[\centering]{{\includegraphics[width=7.1cm]{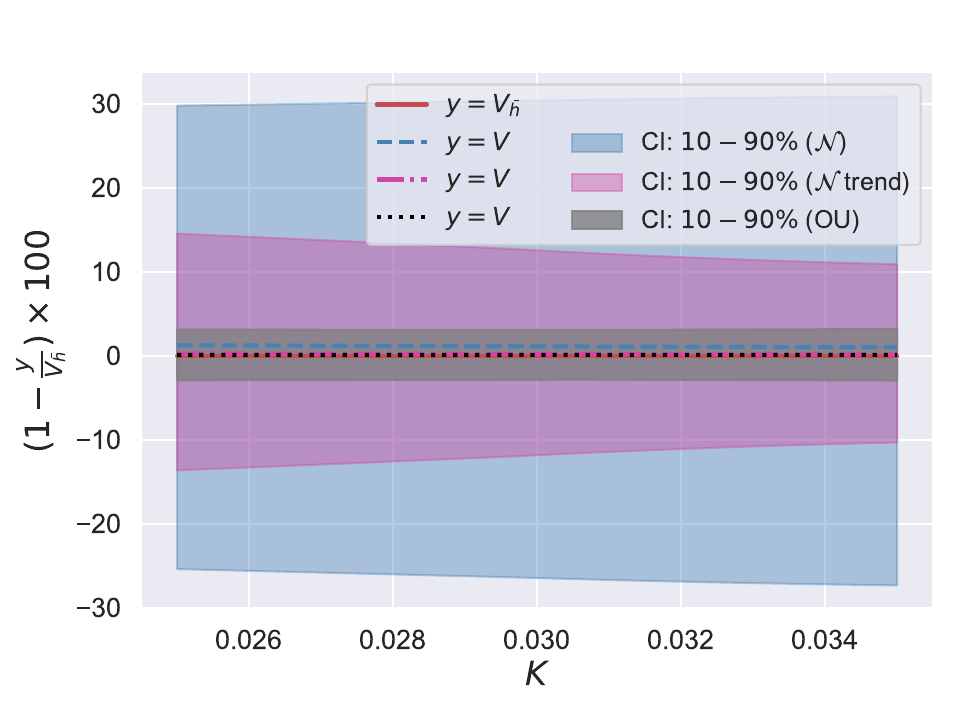} }}%
    ~\hspace{-.2cm}
    \subfloat[\centering]{{\includegraphics[width=7.1cm]{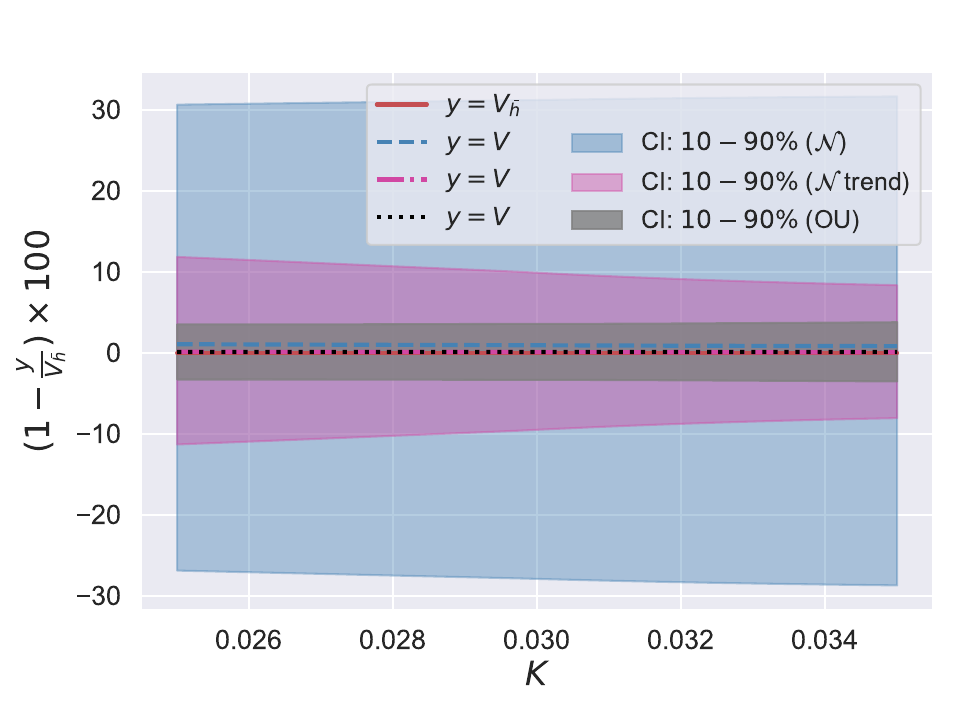} }}
    \caption{\footnotesize Top: Relative price mismatch for bullet (left) and linear (right) amortization schemes. Bottom: Corresponding $10-90\%$ confidence intervals for different HM realizations (bullet on the left and linear on the right).}
    \label{fig: RelativePricesWithCOnfidenceIntervals}%
\end{figure}

\begin{figure}[t!]
    \centering
    \subfloat[\centering]{{\includegraphics[width=7.1cm]{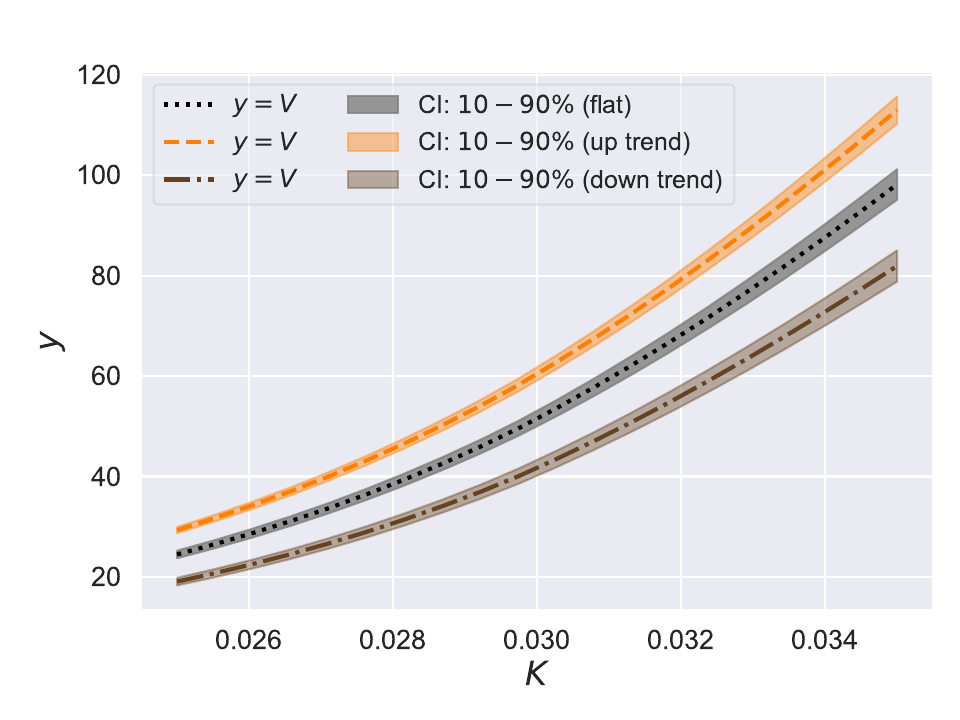} }}%
    ~\hspace{-.2cm}
    \subfloat[\centering]{{\includegraphics[width=7.1cm]{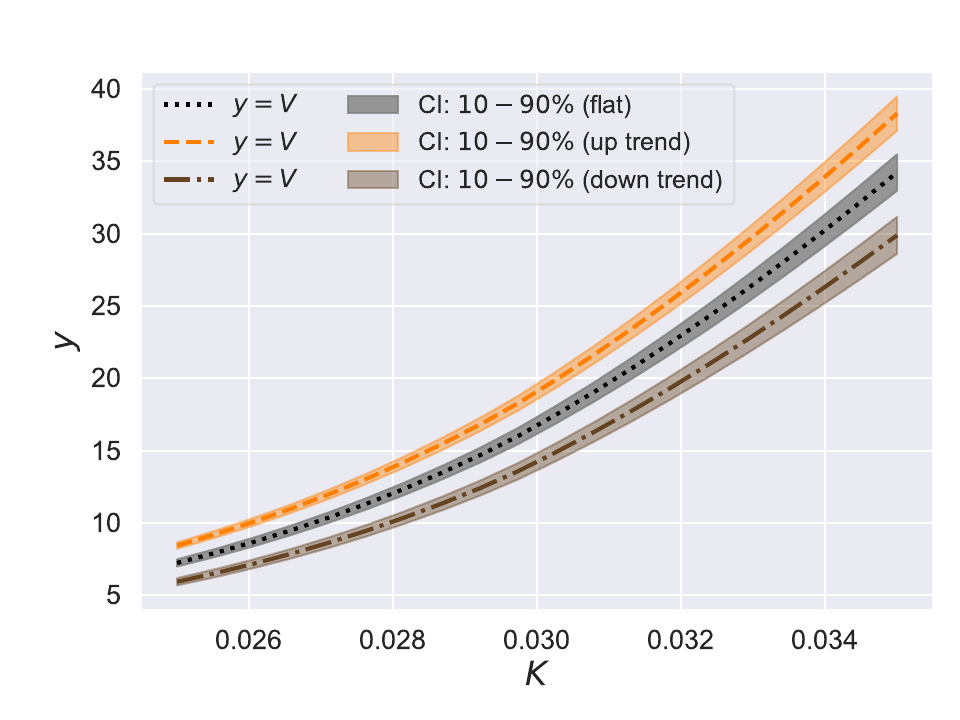} }}
    \caption{\footnotesize Absolute value EPOR (with 10\%-90\% CI) under OU housing market activity with flat, increasing, and decreasing time-dependent long-term mean, $\theta(t)$. Left: Bullet. Right: Linear.}
    \label{fig: PricesDifferentTrendsHM}%
\end{figure}

This section presents the numerical findings from the methodology described in the previous sections.

\subsection{Pricing and impact of the stochastic HM model assumption}

We develop our experiments assuming a flat YC calibrated on market swaps with a par rate $K=0.03$. In the time window $\T=[t_0,T^*]$, $t_0=0$, $T^*=10$, two different mortgage notional amortization schemes are considered: interest-only mortgages (also called ``bullet'') and linear amortization mortgages (referred as ``linear''). The first amortization scheme assumes a constant notional over time, with full redemption of the total debt at the end of the mortgage tenor, $T^*$, and no repayments over the life of the mortgage. In contrast, the latter assumes constant notional repayments over time leading to a zero outstanding notional debt at the final time, $T^*$. In both cases, we consider an initial notional $N(t_0)=10000$, and the results are reported in \emph{basis points} (bps) of the initial notional. The two amortization schemes are illustrated in \Cref{fig: NotionalProfiles}a considering a yearly repayment frequency. The crossed blue and dotted red lines refer to the bullet and the linear contracts, respectively. With the fixed rate selected as the par rate, i.e. $K=0.03$, in \Cref{fig: NotionalProfiles}b, we show the amount of interest paid at every payment date, computed as $KN(t_{j-1})\Delta t_j$, with $\Delta t_j=t_j - t_{j-1}$. A bullet mortgage pays a constant amount of interest over time, while the interest for a linear mortgage decreases following the outstanding notional decrease. The mismatch in interest payments entails that, generally, the EPOR on a bullet has a higher price than the EPOR on a linear mortgage.

For bullet and linear amortization schemes, we compute the EPOR value for fixed rates $K\in[2.5\%,3.5\%]$ considering different models for the stochastic housing market activity, $h$. A significant effect -- about $1\%$ relative difference, see the dashed blue lines in Figures \ref{fig: RelativePricesWithCOnfidenceIntervals}a and \ref{fig: RelativePricesWithCOnfidenceIntervals}b -- is observed only when $h$ is considered normally distributed initially and then flat over the life of the mortgage, as it was assumed in \Cref{fig: DensityDistributions}a and in \Cref{exm: FlatStochasticHM}. The normal distribution is calibrated in \Cref{sec: DataCalibration}, and its density is represented in \Cref{fig: MLEHMtoRelocProb}a as a dashed-dotted blue line. In the other two cases, i.e. when the $h$ grows linearly from the average housing market activity to a final normally distributed activity, as it was assumed in \Cref{fig: DensityDistributions}b, and when $h$ is generated from an Ornstein-Uhlenbeck process, calibrated on the time series in \Cref{fig: MLEHMtoRelocProb}b, with mean the average housing market activity, the relative impact is not so significant -- lower than $0.3\%$, see the dashed-dotted magenta and the dotted black lines in Figures \ref{fig: RelativePricesWithCOnfidenceIntervals}a and \ref{fig: RelativePricesWithCOnfidenceIntervals}b. More interesting, however, is the effect of stochastic $h$ on the EPOR price variance. In Figures \ref{fig: RelativePricesWithCOnfidenceIntervals}c and \ref{fig: RelativePricesWithCOnfidenceIntervals}d, we illustrate the EPOR price ranges between the $10\%$ and $90\%$ quantiles (we call them ``confidence intervals'' (CIs)). Even when the EPOR price is not significantly affected by the stochastic housing market model assumption, the price might show a significant variation, up to several percentage points, represented by the magenta and gray shadows in Figures \ref{fig: RelativePricesWithCOnfidenceIntervals}c and \ref{fig: RelativePricesWithCOnfidenceIntervals}d. Hence, the assumption of stochastic $h$ is relevant to assess the variability in the EPOR price given different housing market activity realizations.

\begin{table}[b]
    \caption{\footnotesize Hedging swaption maturities $\TTT$, ranges $\RRR$, and optimal weights $w^*$ for the different hedging strategies applied to the bullet EPOR.}
    \centering
    \begin{tabular}{ccc|ccc|ccc}
        \toprule
        \multicolumn{3}{c|}{\textit{OpR-MiM}} & \multicolumn{3}{c|}{\textit{FxR-OpM}} & \multicolumn{3}{c}{\textit{FxR-MiM}} \\
        $\TTT$ & $\RRR$ & $w^*$ &$\TTT$ & $\RRR$ & $w^*$ &$\TTT$ & $\RRR$ & $w^*$\\
        \midrule
        0.94 & [0.00, 1.87] & 0.073 & 1.63 & [0.00, 3.33] & 0.124 & 1.67 & [0.00, 3.33] & 0.124 \\
        3.92 & [1.87, 5.97] & 0.136 & 5.53 & [3.33, 6.67] & 0.103 & 5.00 & [3.33, 6.67] & 0.097 \\
        7.98 & [5.97, 10.00] & 0.105 & 8.03 & [6.67, 10.00] & 0.080 & 8.33 & [6.67, 10.00] & 0.095 \\
        \bottomrule
    \end{tabular}
    \label{tab: HedgeBullet}
\end{table}

The major effect on the EPOR price is the level of the housing market activity. Hence, it is crucial to include in the pricing model an accurate forecast of potential future HM trends. To illustrate this fact, we compare the prices obtained using an Ornstein-Uhlenbeck model to generate $h$ when different time-varying means are considered. Specifically, we observe the three cases of a flat, increasing, and decreasing HM activityOrnstein-Uhlenbeck mean. In Figures \ref{fig: PricesDifferentTrendsHM}a and \ref{fig: PricesDifferentTrendsHM}b are displayed -- for bullet and linear mortgages, respectively -- the EPOR prices under the three HM trend assumptions with the corresponding $10\%-90\%$ CIs. The flat, increasing, and decreasing trends are represented in black, orange, and brown, respectively.

\subsection{Hedging of the EPOR exposure}

This section is dedicated to the hedging numerical experiments. We follow the theory of \Cref{sec: Hedging} and we report our numerical findings.

\subsubsection{Explainable hedge}

\begin{table}[t]
    \centering
        \caption{\footnotesize Hedging swaption maturities $\TTT$, ranges $\RRR$, and optimal weights $w^*$ for the different hedging strategies applied to the linear EPOR.}
    \begin{tabular}{ccc|ccc|ccc}
        \toprule
        \multicolumn{3}{c|}{\textit{OpR-MiM}} & \multicolumn{3}{c|}{\textit{FxR-OpM}} & \multicolumn{3}{c}{\textit{FxR-MiM}} \\
        $\TTT$ & $\RRR$ & $w^*$ &$\TTT$ & $\RRR$ & $w^*$ &$\TTT$ & $\RRR$ & $w^*$\\
        \midrule
        0.53 & [0.00, 1.05] & 0.040 & 1.28 & [0.00, 2.00] & 0.073 & 1.00 & [0.00, 2.00] & 0.070 \\
        1.91 & [1.05, 2.77] & 0.062 & 3.28 & [2.00, 4.00] & 0.066 & 3.00 & [2.00, 4.00] & 0.059 \\
        3.71 & [2.77, 4.64] & 0.063 & 4.68 & [4.00, 6.00] & 0.055 & 5.00 & [4.00, 6.00] & 0.050 \\
        5.61 & [4.64, 6.58] & 0.060 & 6.68 & [6.00, 8.00] & 0.051 & 7.00 & [6.00, 8.00] & 0.045 \\
        8.29 & [6.58, 10.00] & 0.124 & 8.48 & [8.00, 10.00] & 0.033 & 9.00 & [8.00, 10.00] & 0.074 \\
        \bottomrule
    \end{tabular}
    \label{tab: HedgeLinear}
\end{table}

\begin{table}[t]
    \centering
        \caption{\footnotesize EPOR value and cost of the different hedging strategies (values are reported in bps).}
    \begin{tabular}{|c|c|c|c|c|}
        \hline
        & \textit{EPOR} & \textit{OpR-MiM} & \textit{FxR-OpM} & \textit{FxR-MiM} \\
        \hline
        {Bullet} & 48.53 & 49.96 & 49.74 & 50.21 \\
        \hline
        {Linear} & 15.66 & 15.95 & 15.91 & 14.06 \\
        \hline
    \end{tabular}
    \label{tab: HedgingCost}
\end{table}

We use a flat YC built on five IRS market quotes for our hedging experiments. The instrument $\S_{yc,i}$, see \Cref{sec: Hedging}, is the swap starting at $t_0$ and ending at $t_i=1,3,5,7,10$ years. As outlined in \Cref{ssec: ExplainableHedge}, we test different strategies for hedging the bullet and linear mortgages. We indicate with \textit{FxR-MiM}, \textit{FxR-OpM}, and \textit{OpR-MiM} the three strategies. \textit{FxR} indicates that the ranges $\RRR$ are fixed a priori (\textit{F}i\textit{x}ed \textit{R}anges) as opposed to \textit{OpR} where the ranges are selected according with some optimality principle (\textit{Op}timal \textit{R}anges). \textit{MiM} indicates that for a given range, the maturity in that range is selected as its mid-point (\textit{Mi}d-point \textit{M}aturity), while \textit{OpM} indicates some optimality criterion has been used to select the maturity (\textit{Op}timal \textit{M}aturity). \textit{FxR-MiM}, \textit{FxR-OpM}, and \textit{OpR-MiM} correspond to \eqref{eq: LocalMinimizationProblem}, \eqref{eq: LocalMinimizationProblemOptimalMaturity}, \eqref{eq: LocalGlobalMinimizationProblemOptimalSubintervals}, respectively.

Aiming at parsimonious hedging strategies, we show the results obtained using three instruments to replicate the bullet EPOR exposure, while five are used for the linear case. Swaptions written on amortizing swaps -- used in the linear case -- show more complex Delta and Gamma profiles, particularly for short maturities. This entails that a satisfactory replication of the EPOR exposure requires a greater number of hedging instruments.

\Cref{tab: HedgeBullet} reports maturities, ranges, and weights from the three hedging strategies in the case of the bullet EPOR. Compared to the baseline case \textit{FxR-MiM}, where the ranges are taken equally spaced in $[t_0,T^*]$ and the maturities of the swaptions are the mid-points of every range, improved strategies are obtained by optimally selecting the swaption maturities (\textit{FxR-OpM}) or by optimally selecting the ranges (\textit{OpR-MiM}). Particularly, notice that a range length is convenient to approximate the Greeks in the first subinterval ($[0.00, 1.87]$ instead of $[0.00, 3.33]$) because of the convexity of the Gamma profile components in the given range. Similarly, in \Cref{tab: HedgeLinear}, the results regarding the linear EPOR hedge are reported. From \Cref{tab: HedgingCost}, we observe that the cost of the hedging strategies is consistent with the EPOR value. Prepayment in the case of a bullet mortgage bears higher risk -- reflected by a higher value and cost of hedging -- when compared to the linear counterpart.

\begin{figure}[b]
    \centering
    \subfloat[\centering]{{\includegraphics[width=7.1cm]{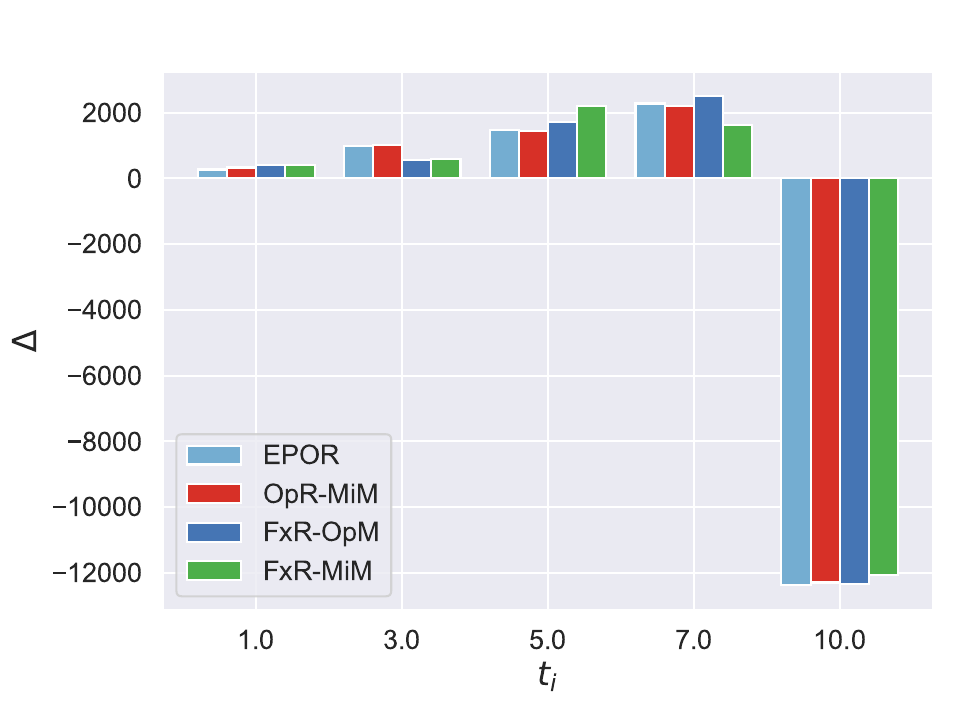} }}%
    ~\hspace{-.2cm}
    \subfloat[\centering]{{\includegraphics[width=7.1cm]{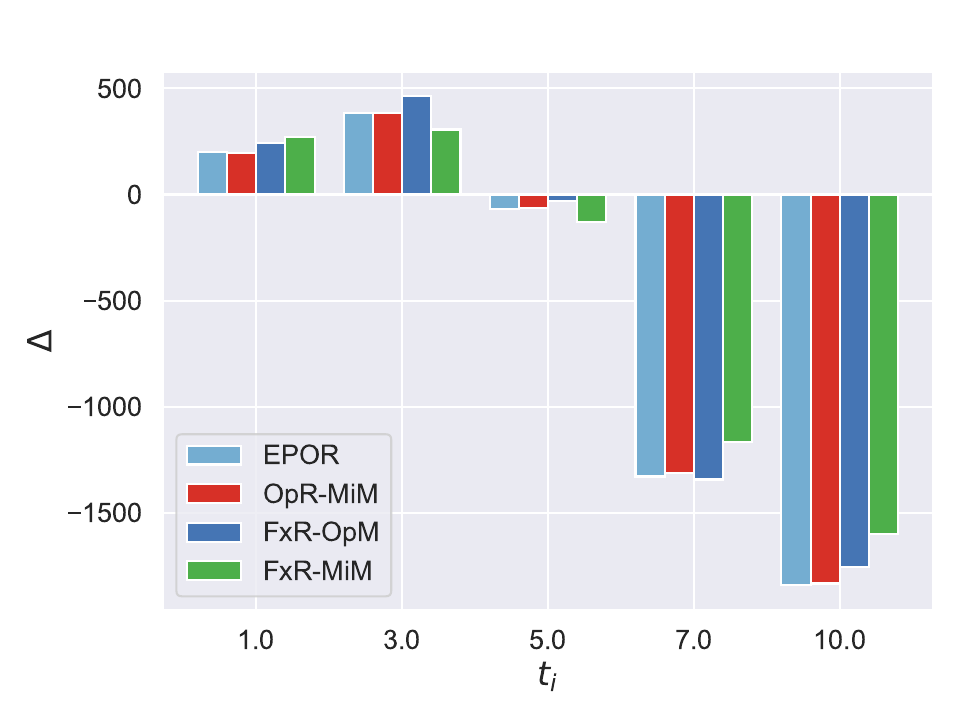} }}
    \caption{\footnotesize Delta profile of EPOR and heading strategies. Left: bullet. Right: linear.}
    \label{fig: DeltaProfiles}%
\end{figure}

\begin{figure}[t]
    \centering
    \includegraphics[width=1.\textwidth]{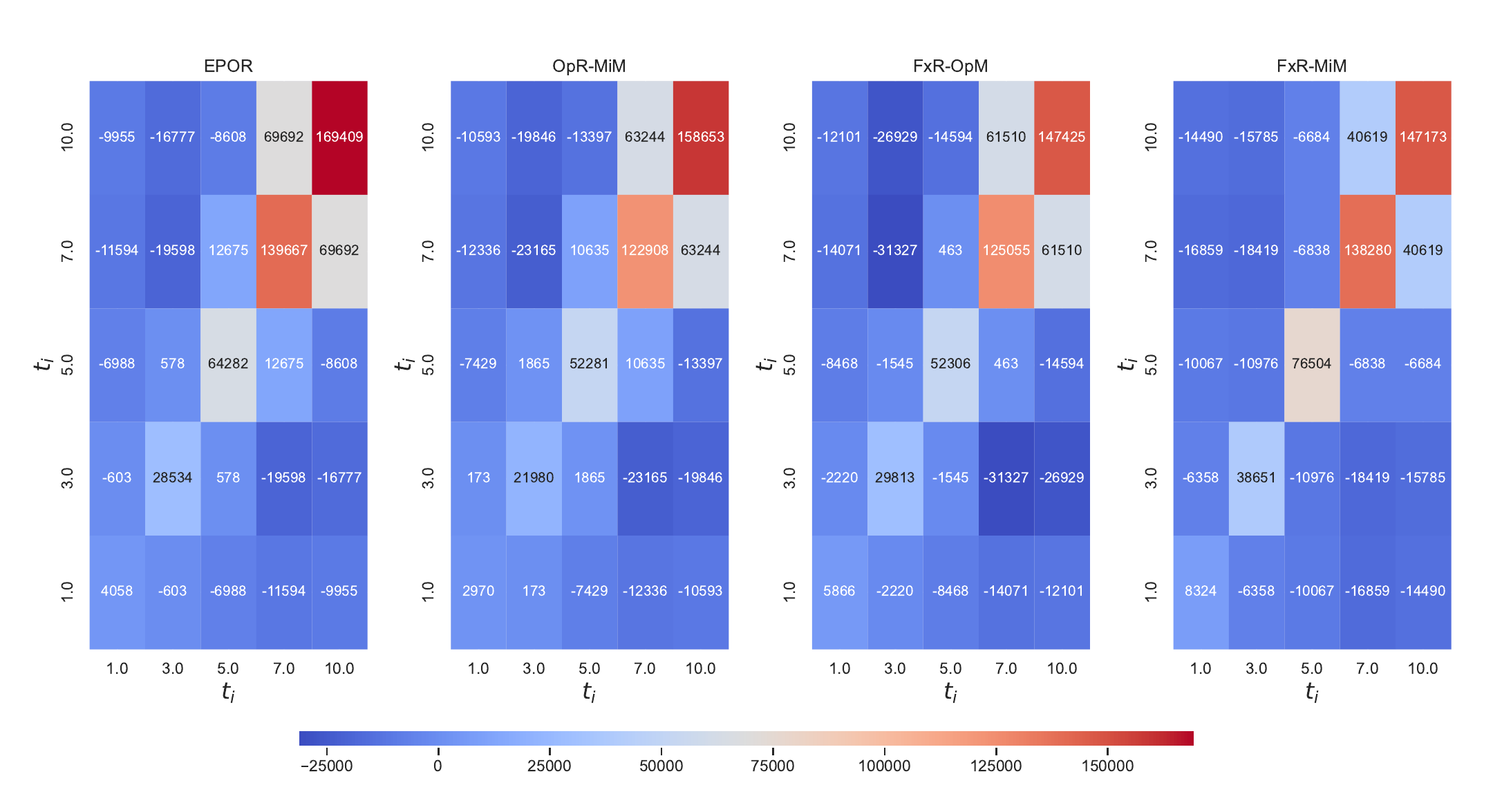}
    \caption{\footnotesize Gamma profile of linear EPOR and hedging strategies.}
    \label{fig: GammaProfiles}
\end{figure}

In \Cref{fig: DeltaProfiles}, the EPOR Delta profiles for the bullet (left) and the linear (right) EPOR are depicted. Because of the amortization schedule, the bullet EPOR Delta is much more significant than the linear EPOR Delta. Furthermore, for the bullet EPOR the Delta is ``concentrated'' on the final maturity $T^*$ (resembling the Delta profile of a vanilla swaption), whereas for the linear case, the profile is more uniform. This is consistent with the amortization type: the linear EPOR resembles an amortizing swaption. Next to the EPOR Delta profile -- represented as a light blue column, we report the Delta profile of the three hedging strategies in red, blue, and green for \textit{OpR-MiM}, \textit{FxR-OpM}, and \textit{FxR-MiM}, respectively. All the strategies are capable of hedging the majority of the EPOR exposure. However, \textit{OpR-MiM} is the strategy that performs best, while \textit{FxR-MiM} often cannot perfectly capture the EPOR Delta, in particular not in the linear case. Such a finding suggests that an optimal selection of the ranges is important to achieve an accurate, yet parsimonious, hedge.

In \Cref{fig: GammaProfiles}, the heatmap of the Gamma profile for the linear EPOR is reported on the left side. The most significant components correspond to the main diagonal of the Gamma profile. Specifically, a more significant Gamma corresponds to a greater shocked maturity (upper right corner). The three heatmaps next to the EPOR Gamma profile are the hedge Gamma for the different hedging strategies. As we observed for the Delta profile in \Cref{fig: DeltaProfiles}b, \textit{FxR-MiM} performs the worst, an improvement is achieved with strategy \textit{FxR-OpM} and the best performance is obtained with \textit{OpR-MiM}. A similar result is found in the case of a bullet EPOR, even though such a case is of less interest since the Gamma is almost exclusively ``concentrated'' in the upper right entry of the Gamma profile.


\subsubsection{Actuarial hedge}

We repeat the hedging experiments following the strategy delineated in \Cref{ssec: ActuarialHedge}. We refer to the hedging strategy as \textit{Eigen}. For illustration purposes, we consider a YC calibrated on three swaps starting at $t_0$ and ending after $1,4,10$ years, respectively. To ensure precise Delta profile matching, in this experiment, we use six swaptions selected solving \eqref{eq: LocalGlobalMinimizationProblemOptimalSubintervals} (i.e., resulting from the \textit{OpR-MiM} strategy in the previous section). To assess the quality of the hedging strategy we observe how a 50 bps shock, in absolute value, on the calibrating swap market quotes affects the position value for the strategy \textit{Eigen}, as opposed to the strategy \textit{OpR-MiM}. We include all the combinations of positive and negative 25 bps shocks on different market quotes. The hyperparameter of the optimization problem is taken as $k_{eig}=3,1$, for the bullet and linear cases, respectively.


\begin{table}[b!]
    \caption{\footnotesize $\ES_{1\%}[\Delta V_{h,shock}]$ at the $1\%$ level and $\P[\Delta V_{h,shock}<0]$ for some of the most relevant shocks for linear EPOR considering the two strategies \textit{OpR-MiM} and \textit{Eigen}.}
    \centering
    \begin{tabular}{c|c|ccccc}
        \toprule
        \multicolumn{2}{c}{} & \multicolumn{5}{c}{Shocks} \\
        \multicolumn{2}{c}{} & [0 0 50] & [0 0 -50] & [0 25 25] & [0 -25 25] & [0 -25 -25] \\
        \midrule
        \multirow{2}{*}{$\ES_{1\%}[\Delta V_{h,shock}]$} 
          & \textit{OpR-MiM}  & -0.618 & -1.032 & -0.412 & -0.311 & -0.609 \\
          & \textit{Eigen}    & -0.266 & -0.642 & -0.283 & -0.223 & -0.440 \\
        \midrule
        \multirow{2}{*}{$\P[\Delta V_{h,shock}<0]$} 
          & \textit{OpR-MiM}  & 74.1\% & 68.3\% & 73.4\% & 61.2\% & 73.4\% \\
          & \textit{Eigen}    & 11.4\% & 24.7\% & 34.6\% & 28.8\% & 40.0\% \\
        \bottomrule
    \end{tabular}
    \label{tab: HedgeEigenLinear}
\end{table}

In \Cref{tab: HedgeEigenLinear} and \Cref{fig: HedgeEigenLinear}, we present the result of the strategy \textit{Eigen} compared to the strategy \textit{OpR-MiM} for the linear EPOR case. The quantity of interest is defined as:
\begin{equation}
\label{eq: ShockedValue}
    \Delta V_{h,shock}=\Big(\sum_j w_j \S_{j, shock} - V_{h,shock}\Big) - \Big(\sum_j w_j \S_{j} - V_{h}\Big),
\end{equation}
where $\S_j$ and $V_h$ are the hedging instrument value and the EPOR value (given HM $h$), respectively, and $\S_{j,shock}$ and $V_{h,shock}$ are the corresponding values given a YC shock. $\Delta V_{h,shock}$ is the change in value of the total position, i.e. the EPOR combined with the hedge, when a certain shock is observed.

\Cref{tab: HedgeEigenLinear} reports the expected shortfall of $\Delta V_{h,shock}$ in \eqref{eq: ShockedValue} at the level $1\%$ and the probability of incurring in a loss, $\P[\Delta V_{h,shock}<0]$, for the most relevant shocks. Compared to the strategy \textit{OpR-MiM}, the strategy \textit{Eigen} significantly reduces the potential losses, summarized by the expected shortfall, occurring when averse HM scenarios are observed. Similarly, the probability of observing a negative $\Delta V_{h,shock}<0$ is significantly reduced. \Cref{fig: HedgeEigenLinear} shows the histograms of the total position variations for each YC shock, for the linear amortization scheme. 
The blue histograms represent the variation observed when the hedging strategy \textit{OpR-MiM} is used, while the orange histograms correspond to the strategy \textit{Eigen} obtained by penalizing negative eigenvalues, see \eqref{eq: MinimizationEigen}. Hedging strategy \textit{OpR-MiM} leads to a change in the value distribution centered around zero. This effect is explained by the fact that such a hedging strategy aims to match the average Delta and Gamma of the EPOR exposure. Hence, depending on the shock sign, both positive and negative value changes are observed. By penalizing the negative eigenvalues, we achieve a more robust hedge that reduces negative changes in value in favor of positive ones. We highlight with a green background the shocks where \textit{Eigen} performs better than \textit{OpR-MiM}. A red background is used otherwise.

\begin{figure}[t!]
    \centering
    \includegraphics[width=1.\textwidth]{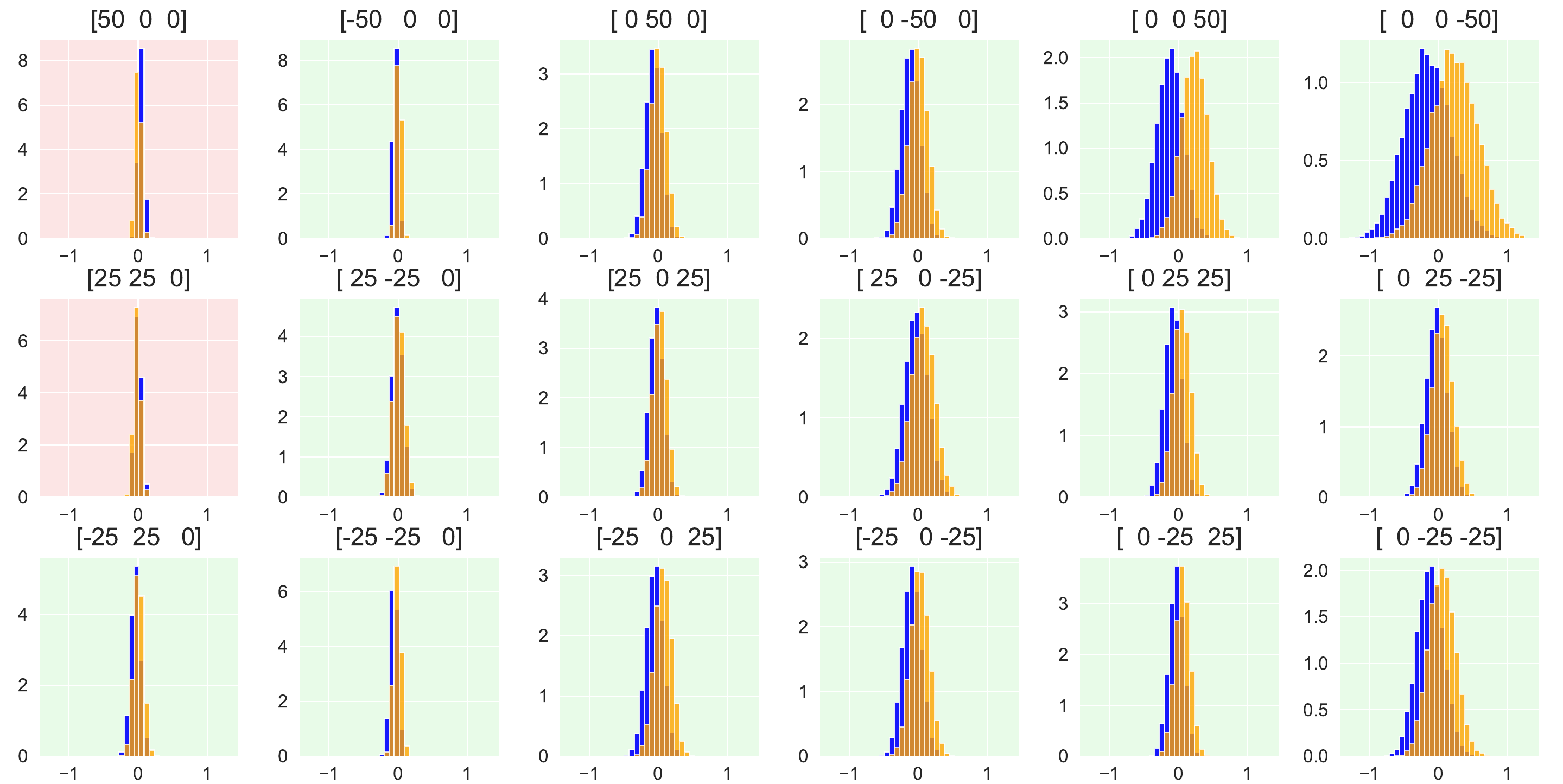}
    \caption{\footnotesize $\Delta V_{h,shock}$ distribution for different shocks for the linear amortization scheme.}
    \label{fig: HedgeEigenLinear}
\end{figure}

\begin{table}[b!]
    \centering
        \caption{\footnotesize Cost of the different \textit{OpR-MiM} and \textit{Eigen} hedging strategies and relative difference.}
    \begin{tabular}{|c|c|c|c|}
        \hline
        & \textit{OpR-MiM} & \textit{Eigen} & \textit{Difference} \\
        \hline
        {Bullet} & 49.17 & 52.90 & 7.59\% \\
        \hline
        {Linear} & 15.93 & 17.34 & 8.85\% \\
        \hline
    \end{tabular}
    \label{tab: HedgingCostEigen}
\end{table}

For both bullet and linear contracts, most of the shocks show a lower negative variation when \textit{Eigen} is preferred to \textit{OpR-MiM}. The few exceptions can be explained by the fact that the optimization in \eqref{eq: MinimizationEigen} prioritizes the eigenvalues bearing the highest concavity, possibly leading to a worse performance along the other, less impactful, directions. However, the strategy \textit{Eigen} successfully improves the hedge performance for the most significant shocks, while only slightly deteriorating the results for the less significant shocks. Compared to the strategy \textit{OpR-MiM}, the investment required in \textit{Eigen} is higher (see \Cref{tab: HedgingCostEigen}), and in general, it increases the more weight $k_{eig}$ is increased.

\section{Conclusion}
\label{sec: Conclusion}

We have proposed a framework for the valuation of the prepayment option embedded in mortgages when the driver of prepayment is the borrower's relocation to a new house (EPOR). We have introduced a stochastic housing market activity and assessed the impact of such a model choice on the EPOR value. We provided an efficient and intuitive pricing formula for the EPOR value corresponding to the price of an EU-type swaption with stochastic maturity. We investigated the theoretical rationale of the mismatch between the stochastic housing market activity and a baseline ``deterministic'' model. The difference is explained by the covariance of the housing market activity risk factor. The exposure generated by the EPOR is hedged using tradable market instruments. We propose strategies that allow for parsimonious, explainable hedges, while closely fitting the EPOR Delta and Gamma profiles. Because of the non-hedgeable nature of the housing market activity risk factor, a more robust, actuarial hedge is suggested. We showed that increased investment in tradable instruments is sufficient to generate a hedging strategy that reduces the potential losses due to extreme housing market activity realizations when compared to a standard hedge.

\section*{Acknowledgments}

L. P. wants to thank the Rabobank ``squads'' \emph{Treasury Analytic and Innovation} and \emph{Prepayment} for their help in developing this research. Special thanks go to Dr. Francois Nissen and Prof. Antoon Pelsser for fruitful discussions and inspiration, and to Mr. Pedro Iraburu for the tireless support provided when dealing with the prepayment data.

\bibliography{Settings/Bibliography}

\myappendix
\section{Proof of \texorpdfstring{\Cref{cor: ConvexityAdjustment}}{}}
\label{app: ProofConvexityAdjustment}
\begin{proof}[Proof of \Cref{cor: ConvexityAdjustment}]
    Observe that, for every $T\in\T$, the functional Taylor expansion of $f(h)=f^h$, see \eqref{eq: FunctionalFTau}, centered in $\Bar{h}$ reads:
    \begin{equation}
    \label{eq: TaylorDensityTau}
        f^h(T) = f^{\Bar{h}}(T)+\int_{t_0}^T \nabla_T^{\Bar{h}}(t)\Delta h(t) \d t+\frac{1}{2}\int_{t_0}^T\int_{t_0}^T \H_{T}^{\Bar{h}}(s,t)\Delta h(s)\Delta h(t)\d s \d t+\O\big(||\Delta h||_T^3\big),
    \end{equation}
    for $\nabla_T^{\Bar{h}}$ and $\H_{T}^{\Bar{h}}$ given by:
    \begin{equation*}
        \nabla_T^{\Bar{h}}(t) = \frac{\delta f^h(T)}{\delta h(t)}\bigg|_{h=\Bar{h}},\quad \H_{T}^{\Bar{h}}(s,t)=\frac{\delta^2 f^h(T)}{\delta h(s) \delta h(t)}\bigg|_{h=\Bar{h}},\quad s,t\in[t_0,T].
    \end{equation*}
    The relationship in \Cref{eq: EPORValueApproximationCOnvexityAdj,eq: NonlinearAdjustmentTerm} follows by substitution of \eqref{eq: TaylorDensityTau} into \Cref{eq: ValueEPORIntegral}, while noticing that $\overline{\Delta h(t)}=\E^\Q\big[\Delta h(t) \big|\F_h(t_0)\big]=0$, for every $t\in[t_0,T]$. The quantity in \Cref{eq: HessianConvexityAdjustment} is obtained by direct derivation. This concludes the proof. 
\end{proof}

\section{Derivation for the discrete approximation of the density Hessian}
\label{app: DiscretizedHessian}
For a given $T\in\T$, we consider the time partition $t_k=t_0 + k\Delta t$, $k=0,\dots,K$ and $\Delta t=(T-t_0)/K$ for some positive integer $K$. We define $\h=[h_1,\dots,h_K]$ with $h_k=h(t_k)$, and we assume $f(h)$ in \eqref{eq: FunctionalFTau} to be approximated by:
\begin{equation*}
    f(h)\approx f(\h)= \lambda(h_K)\e^{-\Sigma(\h)}
\end{equation*}
with $\lambda$ as given in \eqref{eq: RelocationTimeTau} and $\Sigma(\h)$ defined by:
\begin{equation*}
    \Sigma(\h)=\frac12\sum_{k=1}^K(\lambda(h_k)+\lambda(h_{k+1}))\Delta t=\sum_{k=0}^{K} \lambda(h_k)\Delta t-\frac12(\lambda(h_0)+\lambda(h_K))\Delta t.
\end{equation*}
Indicate by $\Sigma'_i(\h)=\frac{\partial \Sigma(\h)}{\partial h_i}$ and $\Sigma''_{i,j}(\h)=\frac{\partial^2 \Sigma(\h)}{\partial h_i \partial h_j}$, then:
\begin{equation*}
    \Sigma'_i(\h) = \lambda'(h_i)\Delta t\1_i,\qquad
    \Sigma''_{i,j}(\h) = \lambda''(h_i)\Delta t\1\{i=j\}\1_i,
\end{equation*}
where $\1\{\cdot\}$ is the indicator function and $\1_i=1-\frac12(\1\{i=0\}+\1\{i=K\})$.
The first derivative $\frac{\partial f}{\partial h_i}$ is given by:
\begin{equation*}
    \frac{\partial f}{\partial h_i}=\Big(-\Sigma'_i(\h) + \frac{\lambda'(h_K)}{\lambda(h_K)}\1\{i=K\}\Big) f(\h).
\end{equation*}
The Hessian represented in \Cref{fig: HessianHeatmapAndIntegral}a is given by:
\begin{align*}
    \frac{\partial^2 f}{\partial h_i\partial h_j}&= f(\h)\Big(-\Sigma'_j(\h) + \frac{\lambda'(h_K)}{\lambda(h_K)}\1\{j=K\}\Big)\Big(-\Sigma'_i(\h) + \frac{\lambda'(h_K)}{\lambda(h_K)}\1\{i=K\}\Big)\\
    &\qquad\qquad+ f(\h)\frac{\partial}{\partial h_j}\bigg(-\Sigma'_i(\h) + \frac{\lambda'(h_K)}{\lambda(h_K)}\1\{i=K\}\bigg)\\
    &=f(\h)\Bigg\{\Sigma'_j(\h)\Sigma'_i(\h)-\Sigma''_{i,j}(\h)+\frac{\lambda''(h_K)}{\lambda(h_K)}\1\{i=j=K\}\\
    &\qquad\qquad\qquad\qquad  -\Big(\Sigma'_j(\h)\1\{i=K\}+\Sigma'_i(\h)\1\{j=K\}\Big)\frac{\lambda'(h_K)}{\lambda(h_K)}\Bigg\}.
\end{align*}
Because of the time discretization employed here, the sum of the components of the Hessian is an approximation of $\int_{t_0}^T\int_{t_0}^T \H^{\Bar{h}}_T(s,t)\d s \d t$, i.e. $\sum_{i,j} \frac{\partial^2 f}{\partial h_i\partial h_j}\approx \int_{t_0}^T\int_{t_0}^T \H^{\Bar{h}}_T(s,t)\d s \d t$, and it is used to generate \Cref{fig: HessianHeatmapAndIntegral}b.

\section{Swap Delta and Gamma profiles}
\label{app: SwapDeltaGamma}
Let us consider a general receiver swap with payment dates $\T_{\tt{p}}$, swap rate $\kappa$, fixed-rate $K$. The value of such a product reads:
\begin{equation*}
    V = A(K-\kappa),
\end{equation*}
where  $\kappa=\kappa(t_0)$ and $A\equiv A(t_0)$ is defined as in \eqref{eq: Annuity}, assuming $N(t)\equiv 1$ for every $t$ (i.e. $A=\sum_{t_j\in\T_{\tt{p}}} \Delta t_j P_j$ for $\Delta t_j = t_j - t_{j-1}$ and $P_j=P(t_0;t_j)$). Let us assume the ZCB curve, $T\mapsto P(t_0;T)$, or -- equivalently -- the YC, is calibrated on a set of $I$ market quotes $\S_{yc}$ (see \Cref{sec: Hedging}).
In this setting, we derive the analytic expression of receiver swap Delta and Gamma profiles w.r.t. the calibrating instruments, $\S_{yc,i}$, $i=1,\dots,I$. Particularly, we obtain the Delta profile:
\begin{equation*}
    \frac{\partial V}{\partial \S_{yc,i}} = \sum_{t_j\in\T_{\tt{p}}} \Delta t_j \frac{\partial P_j}{\partial \S_{yc,i}} (K-\kappa) - A \frac{\partial \kappa}{\partial \S_{yc,i}},\qquad i=1,\dots,I,
\end{equation*}
and the Gamma profile, for $i,\Bar{i}=1,\dots,I$:
\begin{align*}
    \frac{\partial^2 V}{\partial \S_{yc,i} \partial \S_{yc,\Bar{i}}} &= \sum_{t_j\in\T_{\tt{p}}} \Delta t_j \frac{\partial^2 P_j}{\partial \S_{yc,i} \partial \S_{yc,\Bar{i}}} (K-\kappa) - A \frac{\partial^2 \kappa}{\partial \S_{yc,i} \partial \S_{yc,\Bar{i}}}\\
    & -  \sum_{t_j\in\T_{\tt{p}}} \Delta t_j\bigg( \frac{\partial P_j}{\partial \S_{yc,i}} \frac{\partial \kappa}{\partial \S_{yc,\Bar{i}}} + \frac{\partial P_j}{\partial \S_{yc,\Bar{i}}} \frac{\partial \kappa}{\partial \S_{yc,i}} \bigg).
\end{align*}

Typically, at the money (ATM) swaps are selected as calibrating instruments. When we consider the Delta and Gamma profiles of one of the calibrating instruments, we observe that the common intuition that ``linear instruments have vanishing Gamma,'' is not true in this framework. Indeed, let $\S_{yc,i^*}$ be the swap rate of the target calibrated instrument and $V_{yc,i^*}$ its value, then the Delta and Gamma profiles are simplify and read:
\begin{align*}
     \frac{\partial V_{yc,i^*}}{\partial \S_{yc,i}} &=
     \begin{cases}
         - A,&\quad i=i^*,\\
         0,&\quad i\neq i^*,
     \end{cases}\\
     \frac{\partial^2 V_{yc,i^*}}{\partial \S_{yc,i} \partial \S_{yc,\Bar{i}}} &= 
     \begin{cases}
         -  2\sum_{t_j\in\T_{\tt{p}}} \Delta t_j\frac{\partial P_j}{\partial \S_{yc,i^*}},&\quad i=\Bar{i}=i^*,\\
        -  \sum_{t_j\in\T_{\tt{p}}} \Delta t_j\frac{\partial P_j}{\partial \S_{yc,\Bar{i}}},&\quad i=i^*,\Bar{i}\neq i^*,\\
        -  \sum_{t_j\in\T_{\tt{p}}} \Delta t_j\frac{\partial P_j}{\partial \S_{yc,i}},&\quad i\neq i^*,\Bar{i}= i^*,\\
        0,& \quad i\neq i^*,\Bar{i}\neq i^*.
     \end{cases}
\end{align*}
Hence, since in general $\frac{\partial P_j}{\partial \S_{yc,i}}\neq 0$ for at least some $t_j\in\T_{\tt{p}}$, we observe non-vanishing Gamma even for the instruments we used to build the ZCB curve (or the YC). However, the magnitude of the Gamma profile of a swap is orders of magnitude smaller than the Gamma of the corresponding swaption.

\end{document}